\pdfoutput=1 
\documentclass[pdflatex,sn-mathphys-num,iicol]{sn-jnl}

\usepackage{graphicx}%
\usepackage{multirow}%
\usepackage{amsmath,amssymb,amsfonts}%
\usepackage{amsthm}%
\usepackage{mathrsfs}%
\usepackage[title]{appendix}%
\usepackage{xcolor}%
\usepackage{textcomp}%
\usepackage{manyfoot}%
\usepackage{booktabs}%
\usepackage{algorithm}%
\usepackage{algorithmicx}%
\usepackage{algpseudocode}%
\usepackage{listings}%

\usepackage{empheq}
\usepackage{comment}
\usepackage{tikzit}
\usepackage{tikz}
\usetikzlibrary{arrows}
\usetikzlibrary{shapes,backgrounds,shapes.misc}
\usetikzlibrary{matrix,decorations.pathmorphing,decorations.pathreplacing}

\tikzstyle{T}=[fill=black, draw=black, shape=circle, inner sep=0pt, minimum size=1mm]

\tikzstyle{VB}=[fill={rgb,255: red,236; green,255; blue,255}, draw=black, shape=circle, thick, minimum size=6mm, inner sep=0pt, font={\footnotesize}, shading=ball, ball color={{rgb,255: red,236; green,255; blue,255}}]
\tikzstyle{VG}=[fill=green, draw=black, shape=circle, thick, minimum size=6mm, inner sep=0pt, font={\footnotesize}, ball color=green]
\tikzstyle{VR}=[fill=red, draw=black, shape=circle, thick, minimum size=6mm, inner sep=0pt, font={\footnotesize}, ball color=red]
\tikzstyle{VY}=[fill=yellow, draw=black, shape=circle, thick, minimum size=6mm, inner sep=0pt, font={\footnotesize}, ball color=yellow]
\tikzstyle{VP}=[fill=pink, draw=black, shape=circle, thick, minimum size=6mm, inner sep=0pt, font={\footnotesize}, ball color=pink]
\tikzstyle{VS}=[fill=black, draw=black, shape=circle, thick, minimum size=6mm, inner sep=0pt, font={\footnotesize}, ball color=black]
\tikzstyle{VW}=[fill=white, draw=black, shape=circle, thick, minimum size=6mm, inner sep=0pt, font={\footnotesize}, ball color=white]

\tikzstyle{VBsmall}=[fill={rgb,255: red,236; green,255; blue,255}, draw=black, shape=circle, thick, minimum size=2mm, inner sep=0pt, font={\footnotesize}, shading=ball, ball color={{rgb,255: red,236; green,255; blue,255}}]
\tikzstyle{VGsmall}=[fill=green, draw=black, shape=circle, thick, minimum size=2mm, inner sep=0pt, font={\footnotesize}, ball color=green]
\tikzstyle{VRsmall}=[fill=red, draw=black, shape=circle, thick, minimum size=2mm, inner sep=0pt, font={\footnotesize}, ball color=red]
\tikzstyle{VYsmall}=[fill=yellow, draw=black, shape=circle, thick, minimum size=2mm, inner sep=0pt, font={\footnotesize}, ball color=yellow]
\tikzstyle{VPsmall}=[fill=pink, draw=black, shape=circle, thick, minimum size=2mm, inner sep=0pt, font={\footnotesize}, ball color=pink]
\tikzstyle{VSsmall}=[fill=black, draw=black, shape=circle, thick, minimum size=2mm, inner sep=0pt, font={\footnotesize}, ball color=black]
\tikzstyle{VWsmall}=[fill=white, draw=black, shape=circle, thick, minimum size=2mm, inner sep=0pt, font={\footnotesize}, ball color=white]

\tikzstyle{VBfaded}=[fill={rgb,255: red,236; green,255; blue,255}, draw=black, shape=circle, thick, minimum size=2mm, inner sep=0pt, font={\footnotesize}, opacity=0.4]
\tikzstyle{VGfaded}=[fill=green, draw=black, shape=circle, thick, minimum size=2mm, inner sep=0pt, font={\footnotesize}, opacity=0.4]
\tikzstyle{VRfaded}=[fill=red, draw=black, shape=circle, thick, minimum size=2mm, inner sep=0pt, font={\footnotesize}, opacity=0.4]
\tikzstyle{VYfaded}=[fill=yellow, draw=black, shape=circle, thick, minimum size=2mm, inner sep=0pt, font={\footnotesize}, opacity=0.4]
\tikzstyle{VPfaded}=[fill=pink, draw=black, shape=circle, thick, minimum size=2mm, inner sep=0pt, font={\footnotesize}, opacity=0.4]

\tikzstyle{VRBGsquare}=[top color={rgb,255: red,236; green,255; blue,255},  bottom color=green, middle color=red,draw=black, shape=rectangle, thick, minimum size=3mm, inner sep=0pt, font={\footnotesize}]

\tikzstyle{VBGsquare}=[top color={rgb,255: red,236; green,255; blue,255}, bottom color=green, draw=black, shape=rectangle, thick, minimum size=3mm, inner sep=0pt, font={\footnotesize}]
\tikzstyle{VGRsquare}=[top color=green, bottom color=red,draw=black, shape=rectangle, thick, minimum size=3mm, inner sep=0pt, font={\footnotesize}]
\tikzstyle{VRBsquare}=[top color=red, bottom color={rgb,255: red,236; green,255; blue,255}, draw=black, shape=rectangle, thick, minimum size=3mm, inner sep=0pt, font={\footnotesize}]
\tikzstyle{VBYsquare}=[top color={rgb,255: red,236; green,255; blue,255}, bottom color=yellow, draw=black, shape=rectangle, thick, minimum size=3mm, inner sep=0pt, font={\footnotesize}]
\tikzstyle{VGYsquare}=[top color=green, bottom color=yellow,draw=black, shape=rectangle, thick, minimum size=3mm, inner sep=0pt, font={\footnotesize}]
\tikzstyle{VRYsquare}=[top color=red, bottom color=yellow, draw=black, shape=rectangle, thick, minimum size=3mm, inner sep=0pt, font={\footnotesize}]

\tikzstyle{VBsquare}=[fill={rgb,255: red,236; green,255; blue,255}, draw=black, shape=rectangle, thick, minimum size=3mm, inner sep=0pt, font={\footnotesize}]
\tikzstyle{VGsquare}=[fill=green, draw=black, shape=rectangle, thick, minimum size=3mm, inner sep=0pt, font={\footnotesize}]
\tikzstyle{VRsquare}=[fill=red, draw=black, shape=rectangle, thick, minimum size=3mm, inner sep=0pt, font={\footnotesize}]
\tikzstyle{VYsquare}=[fill=yellow, draw=black, shape=rectangle, thick, minimum size=3mm, inner sep=0pt, font={\footnotesize}]

\tikzstyle{VBGfsquare}=[top color={rgb,255: red,236; green,255; blue,255}, bottom color=green, draw=black, shape=rectangle, ultra thick, minimum size=3mm, inner sep=0pt, font={\footnotesize}]
\tikzstyle{VGRfsquare}=[top color=green, bottom color=red,draw=black, shape=rectangle, ultra thick, minimum size=3mm, inner sep=0pt, font={\footnotesize}]
\tikzstyle{VRBfsquare}=[top color=red, bottom color={rgb,255: red,236; green,255; blue,255}, draw=black, shape=rectangle, ultra thick, minimum size=3mm, inner sep=0pt, font={\footnotesize}]
\tikzstyle{VBYfsquare}=[top color={rgb,255: red,236; green,255; blue,255}, bottom color=yellow, draw=black, shape=rectangle, ultra thick, minimum size=3mm, inner sep=0pt, font={\footnotesize}]
\tikzstyle{VGYfsquare}=[top color=green, bottom color=yellow,draw=black, shape=rectangle, ultra thick, minimum size=3mm, inner sep=0pt, font={\footnotesize}]
\tikzstyle{VRYfsquare}=[top color=red, bottom color=yellow, draw=black, shape=rectangle, ultra thick, minimum size=3mm, inner sep=0pt, font={\footnotesize}]

\tikzstyle{VBfsquare}=[fill={rgb,255: red,236; green,255; blue,255}, draw=black, shape=rectangle, ultra thick, minimum size=3mm, inner sep=0pt, font={\footnotesize}]
\tikzstyle{VGfsquare}=[fill=green, draw=black, shape=rectangle, ultra thick, minimum size=3mm, inner sep=0pt, font={\footnotesize}]
\tikzstyle{VRfsquare}=[fill=red, draw=black, shape=rectangle, ultra thick, minimum size=3mm, inner sep=0pt, font={\footnotesize}]
\tikzstyle{VYfsquare}=[fill=yellow, draw=black, shape=rectangle, ultra thick, minimum size=3mm, inner sep=0pt, font={\footnotesize}]

\tikzstyle{VRBGssquare}=[top color={rgb,255: red,236; green,255; blue,255},  bottom color=green, middle color=red,draw=black, shape=rectangle, thick, minimum size=2mm, inner sep=0pt, font={\footnotesize}]

\tikzstyle{VBGssquare}=[top color={rgb,255: red,236; green,255; blue,255}, bottom color=green, draw=black, shape=rectangle, thick, minimum size=2mm, inner sep=0pt, font={\footnotesize}]
\tikzstyle{VGRssquare}=[top color=green, bottom color=red,draw=black, shape=rectangle, thick, minimum size=2mm, inner sep=0pt, font={\footnotesize}]
\tikzstyle{VRBssquare}=[top color=red, bottom color={rgb,255: red,236; green,255; blue,255}, draw=black, shape=rectangle, thick, minimum size=2mm, inner sep=0pt, font={\footnotesize}]
\tikzstyle{VBYssquare}=[top color={rgb,255: red,236; green,255; blue,255}, bottom color=yellow, draw=black, shape=rectangle, thick, minimum size=2mm, inner sep=0pt, font={\footnotesize}]
\tikzstyle{VGYssquare}=[top color=green, bottom color=yellow,draw=black, shape=rectangle, thick, minimum size=2mm, inner sep=0pt, font={\footnotesize}]
\tikzstyle{VRYssquare}=[top color=red, bottom color=yellow, draw=black, shape=rectangle, thick, minimum size=2mm, inner sep=0pt, font={\footnotesize}]

\tikzstyle{VBssquare}=[fill={rgb,255: red,236; green,255; blue,255}, draw=black, shape=rectangle, thick, minimum size=2mm, inner sep=0pt, font={\footnotesize}]
\tikzstyle{VGssquare}=[fill=green, draw=black, shape=rectangle, thick, minimum size=2mm, inner sep=0pt, font={\footnotesize}]
\tikzstyle{VRssquare}=[fill=red, draw=black, shape=rectangle, thick, minimum size=2mm, inner sep=0pt, font={\footnotesize}]
\tikzstyle{VYssquare}=[fill=yellow, draw=black, shape=rectangle, thick, minimum size=2mm, inner sep=0pt, font={\footnotesize}]

\tikzstyle{VBcross}=[shape=cross out, draw=blue, thick, minimum size=1.5mm, inner sep=0pt]
\tikzstyle{VRcross}=[shape=cross out, draw=red, thick, minimum size=1.5mm, inner sep=0pt]

\tikzstyle{VYfadedcircle}=[fill=yellow, shape=cloud, minimum width=14mm, minimum height=14mm, inner sep=0pt, opacity=0.6]
\tikzstyle{VGfadedcircle}=[fill=green, shape=cloud, minimum width=10mm, minimum height=14mm, inner sep=0pt, opacity=0.4]
\tikzstyle{VBfadedcircle}=[fill=blue, shape=cloud, minimum width=14mm, minimum height=14mm, inner sep=0pt, opacity=0.4]
\tikzstyle{VRfadedcircle}=[fill=red, shape=cloud, minimum width=30mm, minimum height=40mm, inner sep=0pt, opacity=0.2]

\tikzstyle{VBpoint}=[shape=circle, fill=blue, draw=blue, thick, minimum size=1mm, inner sep=0pt]
\tikzstyle{VRpoint}=[shape=circle, fill=red, draw=red, thick, minimum size=1mm, inner sep=0pt]

\tikzstyle{Eout}=[->, >=latex]
\tikzstyle{Ein}=[<-, >=latex]
\tikzstyle{Enone}=[-]
\tikzstyle{Einout}=[<->, >=latex]
\tikzstyle{Eoutthick}=[->, >=latex, very thick]
\tikzstyle{Einthick}=[<-, >=latex, very thick]
\tikzstyle{Enonethick}=[-, very thick]
\tikzstyle{Einoutthick}=[<->, >=latex, very thick]
\tikzstyle{Enonedotted}=[-, dash dot]
\tikzstyle{Eoutd}=[->, >=latex, dotted, thin]
\tikzstyle{Eind}=[<-, >=latex, dotted, thin]
\tikzstyle{Enoned}=[-, dotted, thin]
\tikzstyle{Einoutd}=[<->, >=latex, dotted, thin]

\tikzstyle{ERnone}=[-, red]
\tikzstyle{EBnone}=[-, blue]
\tikzstyle{ERnonedotted}=[-, dotted, red, thin]
\tikzstyle{EBnonedotted}=[-, dotted, blue, thin]

\usepackage{siunitx}
\usepackage[capitalize]{cleveref}

\newtheorem{theorem}{Theorem}[section]
\newtheorem{lemma}[theorem]{Lemma}
\newtheorem{corollary}[theorem]{Corollary}
\newtheorem{definition}[theorem]{Definition}

\newcommand{\R}{\mathbb{R}}
\newcommand{\IN}{\mathbb{N}}
\newcommand{\IR}{\mathbb{R}}

\newcommand{\T}{\mathcal{T}}

\newcommand\G{\mathcal{G}}

\newcommand\LL{LL}
\newcommand\LLomega{\LL^{\omega}}

\newcommand\A{\mathcal{A}}
\newcommand\B{\mathcal{B}}

\newcommand\C{\mathcal{C}}

\newcommand\Comega{\mathcal{C}^{\omega}}

\newcommand\V{\mathcal{V}}

\DeclareMathOperator{\In}{In}

\newcommand{\cl}[1]{\overline{#1}}
\newcommand{\subcl}[1]{\widehat{#1}}
\newcommand{\bdtwo}[1]{\bd #1}

\DeclareMathOperator{\bd}{\partial}
\DeclareMathOperator{\bdin}{\partial in}

\DeclareMathOperator{\inte}{Int}

\DeclareMathOperator{\ob}{Ob}
\newcommand{\ho}{HO}

\newcommand{\dunif}{d_{\mathrm{u}}}
\newcommand{\dnonunif}{d_{\mathrm{nu}}}

\newcommand{\noarrow}{\text{---}}

\begin{document}

\title[Topological Characterization of Stabilizing Consensus]{A Topological Characterization of Stabilizing Consensus}


\author*[1]{\fnm{Ulrich} \sur{Schmid}}\email{s@ecs.tuwien.ac.at}

\author[1]{\fnm{Stephan} \sur{Felber}}\email{stephan.felber@tuwien.ac.at}

\author[1,2]{\fnm{Hugo} \sur{Rincon Galeana}}\email{hugorincongaleana@gmail.com}

\affil*[1]{\orgdiv{Embedded Computing Systems Group E191-02}, \orgname{TU Wien}, \orgaddress{\street{Treitlstrasse 1--3}, \city{Vienna}, \postcode{1040}, \country{Austria}}}


\affil[2]{\orgdiv{Internet Architecture and Management}, \orgname{TU Berlin}, \orgaddress{\street{Einsteinufer 17}, \city{Berlin}, \postcode{10587}, \country{Germany}}}



\abstract{
We provide a complete characterization of the solvability/impossibility of 
deterministic stabilizing consensus in any computing model with benign
process and communication faults using point-set topology. Relying on the 
topologies for infinite executions introduced by Nowak, Schmid and Winkler 
(JACM, 2024) for terminating consensus, we prove that semi-open
decision sets and semi-continuous decision functions as introduced 
by Levin (AMM, 1963) are the appropriate means for this characterization:
Unlike the decision functions for terminating consensus, which are
continuous, semi-continuous functions do not require the inverse image of an open set to be open
and hence allow to map a connected space to a disconnected one. 
We also show that multi-valued stabilizing consensus with weak and strong validity 
are equivalent, as is the case for terminating consensus. By applying our
results to (variants of) all the known possibilities/impossibilities for stabilizing 
consensus, we easily provide a topological explanation of these results.
}




\maketitle

\section{Introduction}
\label{sec:intro}

A substantial share of distributed computing research has been devoted to terminating tasks like 
consensus, where every process is given some input value, and must locally
compute some output value and then terminate. Still, there are
also distributed computing problems that cannot be described by such terminating
tasks. Apart from self-stabilizing algorithms \cite{Dol00}, which can recover from arbitrarily
corrupted states, there are also tasks where the processes are allowed to continuously update 
their output values.
Canonical examples are asymptotic consensus \cite{BT89,FNS21:JACM}, stabilizing consensus \cite{AFJ06,CM19:DC}, and
approximate agreement \cite{DLPSW86,CFN15:ICALP}. 
A number of services in practical distributed systems, like
clock synchronization \cite{LL88}, \cite{WS07:DC} and sensor fusion
\cite{BS92}, can be built atop of such non-terminating tasks.

Unlike asymptotic consensus \cite{BT89,Mor05:TAC,HB05,HOT09,HOT11,FNS21:JACM}
and approximate agreement \cite{DLPSW86,Fek90,BT89,Fek94,AAD04,AK00,AK02,CFN15:ICALP,MHVG15,FN18:DISC,FNS21:JACM}, which have been studied in various
computing models and are hence fairly well-understood, not much is known
about stabilizing consensus \cite{AFJ06,CM19:DC,SS21:SSS,FR24:arxiv}.
In stabilizing consensus, processes
need to agree on a common decision value only eventually and do not have
to decide irrevocably, i.e., may change their decisions arbitrarily
before eventually stabilizing to the common value. 
As a straightforward relaxation of terminating consensus,
stabilizing consensus is of particular interest also for
the theory of distributed computing, namely, for studying the
solution power of distributed computing models. For example,
whereas terminating deterministic consensus is impossible to solve in the
synchronous lossy-link model \cite{SW89,SWK09,CGP15}
or in asynchronous systems in the presence of just a single crash fault
\cite{FLP85}, for example, deterministic stabilizing consensus can be
solved in those
models.

More specifically, stabilizing consensus algorithms
for asynchronous systems with
fair-lossy links were given in \cite{AFJ06}, both for crash faults
and byzantine process faults \cite{AFJ06}. 
Solution algorithms for synchronous dynamic
networks controlled by message adversaries \cite{AG13}
were proposed in \cite{CM19:DC}, where Charron-Bost and Moran introduced
the strikingly simple anonymous MinMax algorithms. In particular, as
argued in \cite{FR24:arxiv}, a simple MinMax algorithm can be used to solve stabilizing consensus
in the \emph{lossy-link model} (LL) \cite{SW89,SWK09,CGP15}, where two processes
are connected by a a pair of directed links that may lose at most one message
in every round. 
In \cite{SS21:SSS}, Schmid and Schwarz developed a stabilizing consensus
algorithm for processes with unique ids for the vertex-stable root message
adversary studied in \cite{WSS19:DC}, by stripping-off the termination 
code from a terminating consensus algorithm.

Even less is known regarding impossibility results for
stabilizing consensus.
In \cite{AFJ06}, it was shown that stabilizing consensus is impossibe to
solve in any computing model with byzantine faults and anonymous processes.
In \cite{CM19:DC}, a partitioning argument was used to prove that 
the problem cannot be solved deterministically
in synchronous systems controlled by
a message adversary if the latter does not guarantee a non-empty
kernel, i.e., at least one process that can reach all other processes
(possibly via multiple hops) infinitely often. The, to the best of
our knowledge, first non-obvious impossibility result for deterministic
stabilizing consensus has been established recently by Felber and Rincon Galeana
in \cite{FR24:arxiv}. The authors showed that the problem cannot
be solved in the \emph{delayed lossy-link model} (DLL), where the links
between the two processes may also lose both messages in a round,
provided this happens at most \emph{finitely} often. Note that this
result negatively answered the question raised in \cite{CM19:DC}, namely,
whether a non-empty kernel alone is also \emph{sufficient} for solving stabilizing consensus.

\medskip

In this paper, we provide a complete characterization of the
solvability/impossibility of deterministic stabilizing consensus, in
any model of computation with benign process and communication faults,
using point-set topology as introduced by Alpern and Schneider in \cite{AS84}.
In more detail:
\begin{enumerate}
\item[(1)] Relying on the topologies for infinite
executions introduced by Nowak, Schmid and Winkler 
\cite{NSW24:JACM,NSW19:PODC}, we prove that semi-open
sets and semi-continuous functions as proposed by Levine in \cite{Lev63}
(already\footnote{Since we are not aware of any earlier usage of
Levine's results
in the distributed computing context, our findings constitute another
instance of the utility of purely theoretical (here: mathematical)
research for eventually advancing science in more applied areas.} in 1963) are the appropriate means for this characterization. Unlike the decision 
functions for terminating consensus, which are
continuous, they do not require the inverse image of an open set to be open
and hence allow to map a connected space to a disconnected one.
Since ``offending'' limit points do not need to be excluded from the set of
admissible executions here, this explains why stabilizing consensus is
solvable in models where terminating consensus is impossible.
\item[(2)] We show that, as in the case of terminating consensus, 
stabilizing consensus with weak and
strong validity are equivalent. 
\item[(3)] We demonstrate the power of
our novel topological characterization by applying it to (variants of) all the 
possibility and impossibility results for stabilizing consensus known so far,
also including a new one.
\end{enumerate}

Our paper is organized as follows: In \cref{sec:general:model},
we provide the cornerstones of the generic system model introduced
in \cite{NSW24:JACM} and the variants of stabilizing consensus
considered in this paper. \cref{sec:basictopology} briefly recalls some point-set topology basic
terms and restates the most relevant results from \cite{Lev63}, \cref{sec:terminatingconsensus} 
provides the cornerstones of the topological characterization of terminating 
consensus \cite{NSW24:JACM} needed for our results.
In \cref{sec:characterization}, we provide our topological characterization
for stabilizing consensus, \cref{sec:broadcastability} is devoted to the
equivalence of multi-valued stabilizing consensus with weak and strong validity.
Finally, in \cref{sec:applications}, we apply our topological characterization
to various possibility/impossibility results. Some conclusions in
\cref{sec:conclusions} round-off our paper. 

\section{System Model}\label{sec:general:model}

We use the system model introduced in \cite{NSW24:JACM}, which
targets distributed message passing or shared memory
systems made up of a set of $n$ deterministic processes $\Pi$,
taken 
from $[n]=\{1,\dots,n\}$ for simplicity. 
Individual processes are denoted by letters $p$, $q$, etc.
We assume that all processes start their execution at
the same time $t=0$, in some initial configuration that also
contains some initial values. In order to be able to also model a
process $p$ that becomes \emph{active} at some later time,
we allow $p$ to remain \emph{passive} until some time $t^a_p$.
While being passive, all computing steps $p$ might have to execute (e.g.,
in a lock-step synchronous computing model) are no-operation
steps (that do not change $p$'s state at all),
and all messages possibly arriving at $p$ are lost without a trace.

We restrict our attention to \emph{full-information executions}, in which
processes continuously relay all the information they gathered to
all other processes, and eventually apply some local decision function.
The exchanged information includes the process's initial value, but also, more
importantly, a record of all events (message receptions, shared memory
readings, object invocations, \dots) witnessed by
the process.
Our general system model is hence applicable whenever no constraints are
placed on the size of the local memory and the size of values to be
communicated (e.g., message/shared-register size).
In particular, it is applicable to classical synchronous and asynchronous
message-passing and shared-memory models with benign\footnote{Whereas our
topological modeling would also allow byzantine processes, we deliberately
restricted our attention to benign ones, where processes can at most commit
send and/or receive omissions or crashes: Byzantine processes would require
different validity conditions, the detailed treatment of which 
would blow up our paper considerably, without adding anything substantially new.} 
process and communication faults, see \cite[App.~A]{NSW24:JACM} for details.

Formally, an execution is a sequence of (full-information) configurations that
represent global system states.
For every process $p\in\Pi$, there is an equivalence relation~$\sim_p$ on the
set~$\C$ of configurations---the
{$p$-indistinguishability relation}---indicating whether process~$p$ can
locally distinguish two configurations, i.e., if it
has the same \emph{view} $V_p(C)=V_p(D)$ in~$C$ and~$D$.
In this case we write $C\sim_p D$.
Note that two configurations
that are indistinguishable for all processes need not be equal.
In fact, configurations usually include some state of the communication media
that is not accessible to any process.

In addition to the indistinguishability relations, we assume the existence of a
function $\ob:\C\to2^\Pi$ that specifies the set of \emph{obedient} processes in a
given configuration. Obedient processes must follow the algorithm and satisfy
the task specification. Typically, $\ob(C)$ is the set of non-faulty processes,
but could be extended to also include certain faulty processes, like send omission faulty ones
\cite{BSW11:hyb,WS07:DC}.
Again, the information who is obedient in a given configuration
is usually not accessible to the processes.
We make the restriction that disobedient processes cannot
recover and become obedient again, 
i.e., that $\ob(C) \supseteq \ob(C')$ if~$C'$ is reachable from~$C$.
We extend the obedience function to the set $\Sigma\subseteq\C^\omega$ of
\emph{admissible executions} in a given model by
setting 
$\ob:\Sigma \to 2^\Pi$,
$\ob(\gamma) = \bigcap_{t\geq 0} \ob(C^t)$
where $\gamma = (C^t)_{t\geq 0}$.
Here, $t \in \IN_0 = \IN \cup \{0\}$ denotes a notion
of \emph{global} time that is not accessible to the processes.
Consequently, a process is obedient in an 
execution if it is obedient in all of its
configurations.
We further make the restriction that there is at least one obedient process
in every execution, i.e., that $\ob(\gamma)\neq \emptyset$ for all
$\gamma\in\Sigma$.

We also assume that every process has the possibility to weakly count the steps it
has taken.
Formally, we assume the existence of weak clock functions
$\chi_p:\C\to\IN_0$
such that for every execution 
$\delta = (D^t)_{t\geq0}\in\Sigma$
and every configuration $C\in\C$,
the relation $C\sim_p D^t$ implies $t\geq \chi_p(C)$. 
Additionally, we assume that $\chi_p(D^t)\to\infty$ as $t\to\infty$ for every
execution $\delta\in\Sigma$ and every obedient process $p\in\ob(\delta)$.
$\chi_p$ hence ensures
that a configuration $D^t$ where $p$ has some specific view $V_p(D^t)=V_p(C)$ 
cannot occur before time $t=\chi_p(C)$ in any execution $\delta$.
Our weak clock functions hence allow to model lockstep synchronous
rounds by choosing
$\chi(D^t)=t$ for any execution $\delta = (D^t)_{t\geq0}\in\Sigma$,
but are also suitable for modeling non-lockstep, even asynchronous, executions,
see \cite[App.~A]{NSW24:JACM} for the implementation details.

We assume that the initial configuration $C^0$ of any execution $\gamma=(C^t)_{t\geq0}$ 
contains, for every process $p \in \Pi$,
a local \emph{input value} $I_p=I_p(\gamma) \in \V$ (also called \emph{initial value})
taken from some finite set $\V$, collectively termed \emph{input assignment} $I(\gamma)$. Moreover, every configuration
also contains a local \emph{output value}
$O_p \in \V \cup\{\perp\}$ (also called \emph{decision value}). 
Note that $I_p$ and $O_p$
are locally accessible for $p$ only, i.e., each process only knows
its own initial value (and those it has heard from in the execution),
and $O_p=\perp\not\in \V$ is used to represent the fact that $p$ has not decided 
yet. We make the additional assumption that all input assignments with 
values taken from $\V$ are possible, as formalized in 
the \emph{independent arbitrary input assignment condition} in 
\cref{def:independentinputassumption}.
A \emph{decision algorithm} is a collection of functions 
$\Delta_p:\C \to \V\cup\{\perp\}$ for $p \in \Pi$,
such that $\Delta_p(C) = \Delta_p(D)$ if $C \sim_p D$, i.e., decisions 
depend on local information only and are deterministic. 

We start with the specification of non-uniform and uniform consensus 
\cite{FLP85,CBS04} with weak or strong validity.
Decision functions for terminating consensus have the additional property
$\Delta_p(C') = \Delta_p(C)$ if~$C'$ is reachable from~$C$ and $\Delta_p(C)\neq \perp$, 
i.e., decisions are irrevocable. Since process~$p$ thus has at most one decision value 
in an execution, one can extend the domain of decision functions from configurations to executions by setting
$\Delta_p:\Sigma\to\V\cup\{\perp\}$, $\Delta_p(\gamma) = \lim_{t\to\infty} \Delta_p(C^t)$ 
where $\gamma = (C^t)_{t\geq0}$.
We say that~$p$ has decided value $v\neq\perp$ in configuration~$C$ resp.\
execution~$\gamma$ if $O_p(C)=\Delta_p(C)=v$ resp.\ $O_p(\gamma)=\Delta_p(\gamma)=v$.

\begin{definition}[Non-uniform and uniform terminating consensus]\label{def:consensus}
A \emph{non-uniform terminating consensus} algorithm $\A$ is a decision
algorithm that ensures
the following properties in all of its admissible executions: 
\begin{enumerate}
\item[(T)] Eventually, every obedient process must irrevocably decide.
(Termination)
\item[(A)] If two obedient processes have decided, then their decision values
are equal. (Agreement)
\item[(V)] If the initial values of the processes are all equal to~$v$, then~$v$ is the only possible decision value. (Validity)
\end{enumerate}
In a \emph{strong terminating consensus} algorithm $\A$, weak validity (V) is replaced by
\begin{enumerate}
\item[(SV)] The decision value must be the input value of some process. (Strong Validity)
\end{enumerate}
A \emph{uniform terminating consensus} algorithm $\A$ must ensure (T), (V) or (SV), and
\begin{enumerate}
\item[(UA)] If two processes have decided, then their decision values
are equal. (Uniform Agreement)
\end{enumerate}
\end{definition}

The specification of the corresponding variants of stabilizing 
consensus is given in \cref{def:sconsensus}. Irrevocability and the 
non-decided output value $\perp$ are dropped here, i.e., a process
$p$ may change $O_p \in \V$ finitely many times, albeit it needs to stabilize to 
some common decision value $O_p = \Delta_p(\gamma)=v \in \V$ eventually. 
Again, we will say that~$p$ has decided value $v\neq\perp$ in configuration~$C$ 
resp.\ in execution $\gamma$ if $O_p(C)=\Delta_p(C)=v$ resp.\ $O_p(\gamma)=\Delta_p(\gamma)=v$.

\begin{definition}[Stabilizing consensus]\label{def:sconsensus}
A \emph{stabilizing consensus} algorithm $\A$ is a decision
algorithm that ensures the following properties in all of its admissible executions: 
\begin{enumerate}
\item[(SA)] Eventually, all obedient processes must output the same decision 
value $v \in \V$. (Stabilizing Agreement)
\item[(V)] If the initial values of processes are all equal to~$v$, then~$v$ is the only possible decision value. (Validity)
\end{enumerate}
In a \emph{strong stabilizing consensus} algorithm $\A$, weak validity (V) is replaced by
\begin{enumerate}
\item[(SV)] The decision value must be the input value of some process. (Strong Validity)
\end{enumerate}
\end{definition}

Note that, for terminating consensus, considering Uniform Agreement (UA)
makes sense, since it forces a process who manages to set 
$O_p\neq \perp$ to assign the common decision value when doing so
in that operation---even when $p$ is non-obedient
and thus may crash later on, for example. For stabilizing consensus,
however, there is no meaningful uniform variant of Stabilizing Agreement
(SA): $\perp$ is not used, and it would not make sense to restrict the
choices for $O_p$ for a non-obedient process at any finite time.

For both terminating and stabilizing consensus, (A) resp.\ (SA) 
and the fact that every execution has at least one obedient process allow us to 
define a \emph{global decision function} $\Delta : \Sigma \to \V$, by setting 
$\Delta(\gamma) = \Delta_p(\gamma)$ for an arbitrary process $p$ that is obedient in 
execution~$\gamma$, i.e., $p\in \ob(\gamma)$.

\section{Basics Point-Set-Topology}
\label{sec:basictopology}

We first recall briefly the basic topological notions that are needed for our
exposition; an in-depth treatment can be found in~\cite{Munkres}, for example.

A topology on a set $X$ is a family $\T$ of subsets of $X$
such that $\emptyset \in \T$, $X \in \T$, and $\T$ contains all arbitrary
unions as well as all finite intersections of its members.
We call $X$ endowed with $\T$, often written as $(X, \T)$, a topological
space and the members of $\T$ open sets. For example, if $X$ is endowed
with the discrete topology, every subset $A \subseteq X$ is open.
The complement of an open set is called closed and
sets that are both open and closed, such as $\emptyset$ and $X$ itself, are called clopen.
A topological space is disconnected,
if it contains a nontrivial clopen set, which means that it it can
be partitioned into two disjoint open sets. It is connected
if it is not disconnected.
A space $X$ is called compact if every family of open sets that covers~$X$
contains a finite sub-family that covers~$X$.
A point $x \in X$ in a topological space $X$ is an isolated point if the singleton
set $\{x\}$ is open.

Given $A \subseteq X$, the closure $\cl{A}$ of $A$ is the intersection of all closed
sets containing $A$. The interior $\inte(A)$ of $A$ is the union of all open
sets contained in $A$. The boundary of $A$ is defined as $\bd A=\cl{A}-\inte(A)$.
A set $A$ is called nowhere dense in $X$, if $\inte(\cl{A})=\emptyset$. The closure
of a nowhere dense set, every subset of a nowhere dense set, and
a finite union of nowhere dense sets can be shown to be nowhere dense. 
Nowhere dense sets can also be characterized by the following property:

\begin{lemma}\label{lem:isolationdefinition} A set $A$ of a topological space $X$
is nowhere dense in $X$ if and only if every open set $U$ contains a non-empty open
subset $V\subseteq U$ satisfying $V\cap A = \emptyset$.
\end{lemma}


For a space $X$, if $A \subseteq X$, we call $x$ a limit point of~$A$, if
any open set $U$ containing $x$ intersects $A$ in a point different from 
$x$; this is equivalent to requiring that $x$ belongs to the closure 
of $A \setminus \{ x \}$. The set of all
limit points of $A$ is denoted $A'$, which is also called the derived
set of $A$. It can be shown that 
$\cl{A}=A \cup A' = A \cup \bd A $. Note that $\bd A$ may also contain
points in $A$, namely, ones that are neither interior nor limit points
(in particular, isolated points).

Given a space $(X, \T)$, $Y \subseteq X$ is called a
subspace of $X$ if $Y$ is equipped with the subspace topology
$\{ Y \cap U \mid U \in \T \}$. For a set $A\subseteq Y$, the
closure $\subcl{A}$ of $A$ in $Y$ satisfies $\subcl{A}=
Y \cap \cl{A}$.

A function from space $X$ to space $Y$ is continuous if the pre-image of every
open set in $Y$ is open in $X$.

If~$X$ is a nonempty set, then
we call any function $d:X\times X\to \R_+$ a \emph{distance function} on~$X$.
Many topological spaces are defined by metrics, i.e., symmetric definite
distance functions for which the triangle inequality holds.
For a distance function to define a (potentially non-metrizable) topology
though, no additional assumptions are necessary:
One can define $\T_d \subseteq 2^X$ by setting
$U\in\T$ if and only if
for all $x\in U$ there exists some $\varepsilon > 0$ such that
$B_\varepsilon(x) = 
\{ y \in X \mid d(x,y) < \varepsilon \}
\subseteq
U$.

We recall that the product topology on
a product space $\Pi_{\iota\in I}
X_\iota$ of topological spaces is defined as the coarsest topology such that
all projections $\pi_i:\Pi_{\iota\in I}X_\iota\to X_i$ are continuous.
It turns out that the product topology on the space~$\C^\omega$, i.e., 
infinite sequences of configurations making up executions, is induced
by a distance function whose form is known in a special case that covers
our needs:

\begin{lemma}\label{lem:pseudosemi:product}
Let~$d$ be a distance function on~$X$ that only takes the values~$0$ or~$1$.
Then the product topology of $X^\omega$, where every copy of~$X$ is endowed
with the topology induced by~$d$, is induced by the distance function $d_d$
defined as
\begin{equation}
d_d: X^\omega \times X^\omega \to \R
\quad,\quad
d_d(\gamma,\delta)=
2^{-\inf\{t\geq0\mid d(C^t,D^t) > 0\}}
\label{eq:defminmetric}
\end{equation}
where $\gamma = (C^t)_{t\geq0}$ and $\delta = (D^t)_{t\geq0}$.
\end{lemma}

Equipped with these prerequisites, we can re-state some key results
of Levine's study \cite{Lev63} of semi-open sets and semi-continuous functions,
which will turn out to be a perfect fit for characterizing stabilizing
consensus.

\begin{definition}[Semi-open sets {\cite[Def.~1]{Lev63}}]\label{def:semiopenset} A set A in a topological space $X$ will be termed
semi-open (written s.o.) if and only if there exists an open set $O$ such that $O \subseteq A 
\subseteq \cl{O}$.
\end{definition}

\begin{theorem}[{\cite[Thm.~2]{Lev63}}]\label{thm:arbunionso} Let $\{A_\alpha\}_{\alpha \in \Delta}$ be a
collection of s.o.\  sets in a topological space $X$. Then $\bigcup_{\alpha_\in\Delta} A_\alpha$
is s.o.
\end{theorem}

\begin{theorem}[{\cite[Thm.~6]{Lev63}}] Let $A \subseteq Y \subseteq X$ where $X$ is a topological space and
$Y$ is a subspace. Let $A$ be s.o.\  in $X$. Then $A$ is s.o.\  in $Y$.
\end{theorem}

\begin{lemma}[{\cite[Lem.~1]{Lev63}}]\label{Lem1:Lev} Let $O$ be open in $X$. Then $\bd O = \cl{O}-O$
is nowhere dense in $X$.
\end{lemma}

We introduce the abbreviation $\bdin A$ for boundary points included in $A$,
which allows us to write $S=O \cup \bdin S$ for the open set $O=\inte{S}$ 
guaranteeing $O \subseteq S \subseteq \cl{O}$ in the case of a semi-open set $S$:

\begin{definition}[Included boundary points]\label{def:bdincluded} The set of included boundary
points of $A \subseteq X$ of a topological space $X$ is $\bdin A = A \cap \bd A$.
\end{definition}

\begin{theorem}[{\cite[Thm.~7]{Lev63}}]\label{thm:representationso} Let $A$ be s.o.\  in $X$. Then $A=O \cup B$ where (1) $O$ is open in $X$, (2) $O \cap B = \emptyset$ and (3) $B$ is nowhere dense.
\end{theorem}

In general, the converse of \cref{thm:representationso} is false, i.e., not
every nowhere dense set is an acceptable $B=\bdin A$ for a semi-open set $A$.
The following choice is feasible, though:

\begin{theorem}[{\cite[Thm.~8]{Lev63}}]\label{thm:representationsoother} Let $X$ be a topological space and $A=O \cup B$ where (1) $O\neq\emptyset$ is open in $X$, (2) $A$ is connected and (3) $B' = \emptyset$ where $B'$ is the derived
set [i.e., the set of all limit points] of $B$. Then $A$ is s.o.
\end{theorem}

The following definition introduces semi-continuity as a weaker form of continuity, 
in the sense that a continuous map is also semi-continuous but not necessarily
vice versa:

\begin{definition}[Semi-continuity {\cite[Def.~4]{Lev63}}]\label{DefSemiCont:Lev}  Let $f: X \to X^*$ be
single-valued (continuity not assumed) where $X$ and $X^*$ are topological spaces. Then $f: X \to X^*$
is termed semi-continuous (written s.c.) if and only if, for $O^*$ open in $X^*$, then $f^{-1}(O^*)$
is semi-open in $X$.
\end{definition}

The following theorem is crucial for characterizing semi-continuous functions:

\begin{theorem}[{\cite[Thm.~12]{Lev63}}]\label{Thm12:Lev} 
Let $f: X \to X^*$ be a single-valued function, $X$ and $X^*$ being topological
spaces. Then $f: X \to X^*$ is s.c.\ if and only if for [some open set $O^*$ with]
$f(p)\in O^*$, there exists an $A$ that is semi-open in $X$ such that $p\in A$ and 
$f(A) \subseteq O^*$.
\end{theorem}

For the following theorem, we note that a topological space is a 2nd axiom space 
if it has a countable basis, whereas
a subset of a topological space is of first category if it is a meager set, i.e.,
a countable union of nowhere dense sets.

\begin{theorem}[{\cite[Thm.~13]{Lev63}}]\label{Thm13:Lev} 
Let $f: X \to X^*$ be s.c.\ and $X^*$ a 2nd axiom space. Let $P$ be the
set of points of discontinuity of $f$. Then $P$ is of first category.
\end{theorem}

\section{Topological Characterization of Terminating Consensus}
\label{sec:terminatingconsensus}

In this section, we endow the set of admissible executions $\Sigma$ introduced in Section~\ref{sec:general:model}
with the uniform and non-uniform topologies according to \cite{NSW24:JACM}. Moreover, in order to be able to explain 
the relation between terminating and stabilizing consensus, we also re-state a few core results of the topological characterization of terminating consensus.


Let the \emph{$p$-view distance function}~$d_p$
on the set~$\C$ of configurations for every process $p\in\Pi$ be defined by
\begin{equation}
d_p(C,D)
=
\begin{cases}
0 & \text{if } C \sim_p D \text{ and } p\in\ob(C)\cap\ob(D),\\
1 & \text{otherwise}.
\end{cases}
\end{equation}

Extending this distance function from configurations to executions
according to \cref{lem:pseudosemi:product},
we define the \emph{$p$-view distance function} by
\begin{equation}
d_p:\Sigma \times \Sigma \to \IR_+,\quad
d_p(\gamma,\delta)
=
2^{-\inf\{t\geq0\mid d_p(C^t,D^t) > 0\}}\label{eq:Pviewpseudometric},
\end{equation}
where $\gamma = (C^t)_{t\geq0}$ and $\delta = (D^t)_{t\geq0}$.

\cref{lem:Pseudometricproperties} reveals that it defines a pseudometric,
which differs from a metric by lacking definiteness: There may be executions
$\gamma\neq\delta$ with $d_p(\gamma,\delta) =0$.

\begin{lemma}[Pseudometric $d_p$ {\cite[Lem.~4.3]{NSW24:JACM}}]\label{lem:Pseudometricproperties}
The $p$-view distance function $d_p$ is a pseudometric, i.e., it satisfies:
\begin{align*}
d_p(\gamma,\gamma)&=0 \\
d_p(\gamma,\delta) &=  d_p(\delta,\gamma) \qquad\mbox{(symmetry)}\\
d_p(\beta,\delta) &\leq d_p(\beta,\gamma) + d_p(\gamma,\delta) \qquad\mbox{(triangle inequality)}
\end{align*}
\end{lemma}

The \emph{uniform minimum topology} (abbreviated \emph{uniform topology}) on the set~$\Sigma$ of executions is induced by the distance function
\begin{equation}
\dunif(\gamma, \delta)
=
\min_{p\in\Pi} d_p(\gamma,\delta).\label{eq:dunif}
\end{equation}

Note that $\dunif$ does not necessarily satisfy the triangle inequality (nor definiteness): There may
be executions with $d_{p}(\beta,\gamma)=0$ and $d_{q}(\gamma,\delta)=0$
but $d_{r}(\beta,\delta)>0$ for all $r\in \Pi$. Hence, the topology
on $\Sigma$ induced by $\dunif$ lacks some of the convenient (separation)
properties of metric spaces.

The next lemma shows that the decision function of an algorithm that solves
uniform consensus is always continuous with respect to
the uniform topology.

\begin{lemma}[{\cite[Lem.~4.4]{NSW24:JACM}}]\label{lem:cont:unif:consensus}
Let $\Delta:\Sigma\to\V$ be the decision function of a
uniform consensus algorithm.
Then, $\Delta$ is continuous with respect to 
the uniform distance function~$\dunif$.
\end{lemma}

\begin{definition}[$v$-valent execution]\label{def:vvalentexec}
We call an execution $\gamma_v \in \Sigma$, for $v\in\V$, $v$-valent, if it
starts from an initial configuration where all processes $p\in\Pi$ have the
same initial value $I_p(\gamma_v)=v$.
\end{definition}

\begin{theorem}[Characterization of uniform consensus {\cite[Thm.~5.2]{NSW24:JACM}}]\label{thm:char:unif}
Uniform consensus is solvable if and only if there exists
a partition of the set~$\Sigma$ of admissible executions
into sets $\Sigma_v$, $v\in\V$, such that
the following holds:
\begin{enumerate}
\item Every $\Sigma_v$ is a clopen set in~$\Sigma$ with respect to the
uniform topology induced by~$\dunif$.
\item If execution $\gamma\in\Sigma$ is $v$-valent, then $\gamma \in \Sigma_v$.
\end{enumerate}
\end{theorem}


Whereas the $p$-view distance function given by \cref{eq:Pviewpseudometric}
is also meaningful for non-uniform consensus, this is not the case for the
uniform distance function as defined in \cref{eq:dunif}.
The appropriate
\emph{non-uniform minimum topology} (abbreviated \emph{non-uniform topology}) 
on the set~$\Sigma$ of executions is induced by the distance function
\begin{equation}
\dnonunif(\gamma, \delta)
=
\begin{cases}
\min_{p\in\ob(\gamma)\cap\ob(\delta)} d_p(\gamma,\delta)\\
\qquad\qquad \text{if } \ob(\gamma)\cap\ob(\delta) \neq \emptyset,\\
1\\
\qquad\qquad \text{if } \ob(\gamma)\cap\ob(\delta) = \emptyset.
\end{cases}\label{eq:dnonunif}
\end{equation}

Like for $\dunif$, neither definiteness nor the triangle inequality need to be
satisfied by $\dnonunif$. The resulting non-uniform topology is finer than the uniform topology, however, since
the minimum is taken over the smaller set
$\ob(\gamma)\cap\ob(\delta)\subseteq \Pi$, which means that
$\dunif(\gamma,\delta) \leq \dnonunif(\gamma,\delta)$.
In particular, this implies that every decision function that is continuous with
respect to the uniform topology is also continuous with respect to the 
non-uniform topology.

\begin{lemma}[{\cite[Lem.~4.5]{NSW24:JACM}}]\label{lem:cont:nonunif:consensus}
Let $\Delta:\Sigma\to\V$ be the decision function of a
non-uniform consensus algorithm.
Then, $\Delta$ is continuous with respect to
the non-uniform distance function~$\dnonunif$.
\end{lemma}

\begin{theorem}[Characterization of non-uniform consensus {\cite[Thm.~5.3]{NSW24:JACM}}]\label{thm:char:nonunif}
Non-uniform consensus is solvable if and only if there exists
a partition of the set~$\Sigma$ of admissible executions
into sets $\Sigma_v$, $v\in\V$, such that
the following holds:
\begin{enumerate}
\item Every $\Sigma_v$ is a clopen set in~$\Sigma$ with respect to the
non-uniform topology induced by~$\dnonunif$.
\item If execution $\gamma\in\Sigma$ is $v$-valent, then $\gamma \in \Sigma_v$.
\end{enumerate}
\end{theorem}

\subsection{The importance of limit points for terminating consensus}
\label{sec:limitpoints}

Additional results established in \cite{NSW24:JACM} reveal that certain 
limit points play a crucial role for consensus
solvability. Indeed, since consensus is a continuous function 
(\cref{lem:cont:unif:consensus} and \cref{lem:cont:nonunif:consensus}),
it is not too surprising that consensus 
is impossible if and only if limit points ``between'' different
decision sets $\Sigma_v$ and $\Sigma_w$ are admissible.

Applying the following well-known topological \cref{lem:separation} 
to the findings of \cref{thm:char:unif} resp.\ \cref{thm:char:unif}
yields the limit-based characterization of consensus solvability given in 
\cref{cor:consensusseparation}.

\begin{lemma}[Separation Lemma {\cite[Lemma~23.12]{Munkres}}]\label{lem:separation}
If $Y$ is a subspace of $X$, a separation of $Y$ is
a pair of disjoint nonempty sets $A$ and $B$ whose union is $Y$, neither of
which contains a limit point of the other. The space $Y$ is connected
if and only if there exists no separation of $Y$. Moreover, $A$ and $B$ of
a separation of $Y$ are clopen in $Y$.
\end{lemma}

\begin{theorem}[Separation-based consensus characterization {\cite[Thm.~6.4]{NSW24:JACM}}]\label{cor:consensusseparation}
Uniform resp.\ non-uniform consensus is solvable in a model if and only if
there exists a partition of the set of admissible executions $\Sigma$ into decision sets $\Sigma_v,v\in\V$, such that
the following holds:
\begin{enumerate}
\item No $\Sigma_v$ contains a limit point of any other $\Sigma_w$ w.r.t.\ the uniform resp.\
non-uniform topology in~$\Comega$.
\item Every $v$-valent admissible execution $\gamma_v$ satisfies $\gamma_v\in \Sigma_v$.
\end{enumerate}
If consensus is not solvable, then $\dunif(\Sigma_v,\Sigma_w)=0$ resp.\ $\dnonunif(\Sigma_v,\Sigma_w)=0$ for some $w\neq v$.
\end{theorem}

\cref{cor:consensusseparation} can also be expressed via the exclusion of fair/unfair 
executions as defined in \cite{FG11}:

\begin{definition}[Fair and unfair executions {\cite[Def.~6.6]{NSW24:JACM}}]\label{def:fairunfair}
Consider two executions $\rho, \rho' \in \Comega$
of some consensus algorithm with decision sets $\Sigma_v$, $v\in\V$, in any appropriate topology:
\begin{itemize}
\item $\rho$ is called \emph{fair}, if for some $v,w\neq v 
\in \V$ there are convergent sequences $(\alpha_k) \in \Sigma_v$
and $(\beta_k) \in \Sigma_w$ with
$\alpha_k\to \rho$ and $\beta_k\to \rho$.
\item $\rho$, $\rho'$ are called a pair of \emph{unfair} executions,  
if for some $v,w\neq v 
\in \V$ there are convergent sequences $(\alpha_k) \in \Sigma_v$
with $\alpha_k\to \rho$ and $(\beta_k) \in \Sigma_w$ with $\beta_k\to \rho'$
and $\rho$ and $\rho'$ have distance 0.
\end{itemize}
\end{definition}

\begin{corollary}[Fair/unfair consensus characterization {\cite[Cor.~6.7]{NSW24:JACM}}]\label{cor:consensusimpfair}
Condition (1) in \cref{cor:consensusseparation} is equivalent to requiring
that the decision sets $\Sigma_v$, $\Sigma_w$ for $w\neq v$ neither contain
any fair execution nor any pair $\rho,\rho'$ of unfair executions.
\end{corollary}

An illustration of our limit-based characterizations is provided by Figure~\ref{fig:noncompactMA}.

\tikzset{cross/.style={cross out, draw=black, minimum size=2*(#1-\pgflinewidth), inner sep=0pt, outer sep=0pt},
cross/.default={1pt}}

\begin{figure}[ht]
\centering  
\begin{tikzpicture}[scale=0.8]
\pgfmathsetseed{2}
\draw [blue!40!red, dashed, ultra thick, shift={(0,0)}] 
plot [smooth cycle, tension=1, domain=0:320, samples=18] (\x:{4/3+rand/3}) 
node[left=0cm,above=0.2cm,text=black] {$\Sigma_{\gamma_{0}}\qquad$} ;
\draw [green!50!blue, dashed, ultra thick, shift={(3.1,0)}] 
plot [smooth cycle, tension=1, domain=0:320, samples=18] (\x:{4/3+rand/3}) 
node[left=0cm,above=0.2cm,text=black] {$\Sigma_{\gamma_{1}}\qquad$} ;
\draw [blue!40!red, dashed, ultra thick, shift={(0,3.1)}] 
plot [smooth cycle, tension=1, domain=0:320, samples=18] (\x:{4/3+rand/3}) 
node[left=0cm,above=0.2cm,text=black] {$\Sigma_{\gamma_{0}'}\qquad$} ;
\draw [green!50!blue, dashed, ultra thick, shift={(3.1,3.1)}] 
plot [smooth cycle, tension=1, domain=0:320, samples=18] (\x:{4/3+rand/3}) 
node[left=0cm,above=0.2cm,text=black] {$\Sigma_{\gamma_{1}'}\qquad$} ;
\foreach \i in {0,...,5} {
\draw[fill,color=green!50!blue] ({1.39+1.5^(-\i)},0.3) circle (0.05);
\draw[fill,color=blue!40!red] ({1.39+1.5^(-5)-1.5^(-\i)},0.3) circle (0.05);
}
\draw ({1.36+1.5^(-5)},0.3) node[cross=4pt,black, ultra thick]{};
\end{tikzpicture}%
\caption{Examples of two connected components of the decision sets $\Sigma_0=\Sigma_{\gamma_0}\cup \Sigma_{\gamma_0'}$ and
$\Sigma_1=\Sigma_{\gamma_1}\cup \Sigma_{\gamma_1'}$.
Common limit points (like for $\Sigma_{\gamma_0}$ and $\Sigma_{\gamma_1}$, marked by $\times$) must be excluded by
\cref{cor:consensusseparation}.}
\label{fig:noncompactMA}
\end{figure}
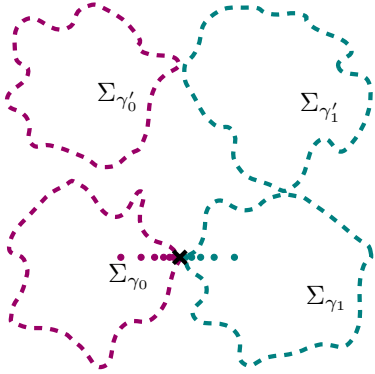

\section{Topological Characterization of Stabilizing Consensus}
\label{sec:characterization}

In this section, we will use the non-uniform topology introduced in \cref{sec:terminatingconsensus} 
in conjunction with the topological results on semi-open sets and semi-continuous functions
developed by Levine in \cite{Lev63} (which we re-stated in \cref{sec:basictopology})
for characterizing the solvability of stabilizing consensus. 
Note that the uniform topology would
not be useful here, since there is no meaningful uniform variant of Stabilizing Agreement (SA)
in \cref{def:sconsensus}.

We already noted at the beginning
of \cref{sec:limitpoints} that the continuity of the decision function of terminating
consensus makes it mandatory that there is no admissible limit point ``between''
two different decision sets, recall the illustration in \cref{fig:noncompactMA}.
After all, a continous function cannot map a connected space to a disconnected
space. For example, since all executions are admissible in the lossy-link model,
resulting in a connected space $\Sigma=\C^\omega$, it is hence impossible to 
solve terminating consensus.

Since stabilizing consensus \emph{can} be solved in the lossy-link model (see \cref{sec:LLmodel}), its decision function
cannot be continuous. Rather, it must be able to map a connected space to a disconnected
one. \cref{lem:Scont:Sconsensus} below will reveal that the decision function of
a \emph{proper} stabilizing consensus algorithm must be \emph{semi-continuous}, in the sense 
of \cref{DefSemiCont:Lev}, and that the decision sets must be \emph{semi-open}, in the sense
of \cref{def:semiopenset}. \emph{Proper} decision functions will 
ensure that no decision set $\Sigma_v$ contains an execution
that is \emph{not} one of its own limit points (but rather a limit point of some other
decision set $\Sigma_w$):

\begin{definition}[Proper stabilizing consensus]\label{def:properSC} A decision
function for solving stabilizing consensus is called \emph{proper},
iff
\begin{enumerate}
\item[(i)] its domain covers independent arbitrary input assignments as defined in 
\cref{def:independentinputassumption}, and
\item[(ii)] none of its induced decision sets $\Sigma_v$, $v\in \V$, contains 
any execution $\gamma \not\in \cl{\inte(\Sigma_v)}$, unless $\gamma$ is an isolated 
point in $\Sigma$ (where $\{\gamma\}$ is open in $\Sigma$).
\end{enumerate}
\end{definition}
Note that an isolated point in $\Sigma$ is also an isolated point
in the subspace $\Sigma_v$.

Whereas \cref{def:properSC} restricts the class of stabilizing consensus algorithms
covered by our topological characterization, we will prove later (\cref{lem:nonproperreduction}) 
that it is easy to turn a decision function
$\Delta'$ of a non-proper stabilizing consensus algorithm that satisfies (i) into 
a decision function $\Delta$ that also satisfies (ii), i.e., is proper.

We first instantiate \cref{Thm12:Lev} in our special context, which exploits
the fact that every singleton set $\{v\}$ for $v\in\V$ is open in the discrete
topology.

\begin{corollary}[Characterization of semi-continuous functions]\label{OurThm12:Lev} 
Let $\Delta: \Sigma \to \V$ (where $\Sigma$ is equipped with the non-uniform topology, 
and $\V$ is equipped with the discrete topology) 
be a single-valued function.
Then $\Delta: \Sigma \to \V$ is semi-continuous if and only if for every execution
$\gamma \in \Sigma$ with $\Delta(\gamma)=v\in\V$, there exists a semi-open set $A$
in $\Sigma$ such that $\gamma\in A$ and $\Delta(A)=v$.
\end{corollary}

\begin{lemma}\label{lem:Scont:Sconsensus}
Let $\Delta:\Sigma\to\V$ be the decision function of a proper
stabilizing consensus algorithm.
Then, $\Delta$ is semi-continuous with respect to 
the 
non-uniform distance function~$\dnonunif$.
\end{lemma} 
\begin{proof}
We use the only-if direction ($\Leftarrow$) of \cref{OurThm12:Lev}. Let $\Sigma_v=\Delta^{-1}(v)$, 
and distinguish 3 cases for $\gamma\in\Sigma_v$ when determining the required semi-open
set $A \subseteq \Sigma_v$ containing $\gamma$:

(1) If $\gamma \in \inte(\Sigma_v)$, there must be some $k\geq 0$ such
that the open ball $B_{2^{-k}}(\gamma)$ with radius $2^{-k}$ (for
distance function $\dnonunif$) satisfies $B_{2^{-k}}(\gamma) \subseteq \Sigma_v$.
We take $A=B_{2^{-k}}(\gamma)$ for \cref{OurThm12:Lev} in this case, which
is obviously open and hence semi-open.

(2) If $\gamma$ is an isolated point in $\Sigma_v$, the singleton set
$\{\gamma\}$ is open in $\Sigma_v$, so we can safely take $A=\{\gamma\}$,
which is again open and hence semi-open.

(3) If $\gamma \in \bdin \Sigma_v$ is an (included) boundary point and hence a
limit point of $\Sigma_v$ (as $\Delta$ is a proper stabilizing consensus decision function),
every open ball $B_{2^{-k}}(\gamma)$, $k\geq 0$, also intersects $\Sigma_v$ in a point 
$\beta_k$ different from $\gamma$. Choosing $\beta_k \in \inte(\Sigma_v)$,
there is some $\ell_k\geq 0$ such that the open ball $B_{2^{-\ell_k}}(\beta_k) \subseteq \inte(\Sigma_v)$.
We thus take $A=\{\gamma\}\cup\bigcup_{k\geq 0}B_{2^{-\ell_k}}(\beta_k)$,
which is semi-open as an arbitrary union of (semi-)open sets with limit point $\gamma$, 
recall \cref{thm:arbunionso}.

Applying the only-if direction of \cref{OurThm12:Lev} proves the asserted semi-continuity
of $\Delta$.
\end{proof}

By applying \cref{Lem1:Lev} to the result of \cref{lem:Scont:Sconsensus} while recalling
\cref{def:properSC}.(ii), we immediately observe that the boundaries of
all decision sets are nowhere dense:

\begin{corollary}[Nowhere dense boundaries]\label{lem:nowheredenselimits} 
Let $\Delta: \Sigma \to \V$ be a proper stabilizing consensus decision function,
and $\Sigma_v$, $v\in\V$ be the corresponding decision sets. Then, 
the boundary $\bd \Sigma_v$, and hence the included boundary $\bdin \Sigma_v$,
of every decision set is nowhere dense.
\end{corollary}

Since the union of the boundaries $\bd \Sigma_v$ of all decision sets
contains the points of discontinuity of $\Delta$,
\cref{lem:nowheredenselimits} is consistent with the result of the following
\cref{lem:nowheredense}, which is a specialization of \cref{Thm13:Lev}:

\begin{theorem}[Nowhere dense discontinuities]\label{lem:nowheredense} 
Let $\Delta: \Sigma \to \V$ be a proper stabilizing consensus decision function. Then, the set $P$
of points of discontinuity of $\Delta$ in $\Sigma$ is nowhere dense.
\end{theorem}
\begin{proof}
Since $\V$ is equipped with the discrete topology, it has a finite
open basis since the singleton sets $\{v\}$, $v\in \V$, are open.
For every point of discontinuity $\gamma\in P$, there exists an open set $V_\gamma=\{v\}$
for some $v\in \V$, such that for some open set $W \subseteq \Sigma$
it holds that $\gamma \in W$ but $\Delta(W) \not\subseteq V_\gamma$.
By \cref{OurThm12:Lev}, there exists a semi-open set $A_\gamma$ with
$\gamma \in A_\gamma$ and $\Delta(A_\gamma) \subseteq \{v\}=V_\gamma$.
By definition, $A_\gamma = O_\gamma \cup B_\gamma$ with $B_\gamma
\subseteq \bd O_\gamma = \cl{O_\gamma}-O_\gamma$, so $\gamma \not\in O_\gamma$
and hence $\gamma\in B_\gamma$. Since $\bd O_\gamma$ is nowhere dense
by \cref{Lem1:Lev}, the subset $B_\gamma$ is also nowhere dense.
As there are only finitely many different sets $V_\gamma$,
it follows that $\bigcup_{\gamma \in P}B_\gamma$ is a finite
union of nowhere dense sets, which is also nowhere dense. Since
$P \subseteq \bigcup_{\gamma \in P}B_\gamma$, this also holds
for $P$.
\end{proof}



The following \cref{thm:Schar:nonunif} provides the topological characterization of
stabilizing consensus solvability in the non-uniform topology:

\begin{theorem}[Characterization of stabilizing consensus]\label{thm:Schar:nonunif}
Stabilizing consensus with weak validity is solvable with a proper decision function
if and only if the set~$\Sigma$ of admissible executions can be split
into disjoint sets $\Sigma_v$, $v\in\V$, with $\Sigma=\bigcup_{v\in\V}\Sigma_v$, such that
the following holds:
\begin{enumerate}
\item[(1)] Every $\Sigma_v$ is semi-open in~$\Sigma$ with respect to the
non-uniform topology induced by~$\dnonunif$.
\item[(2)] If execution $\gamma\in\Sigma$ is $v$-valent, then $\gamma \in \Sigma_v$.
\end{enumerate}
\end{theorem}
\begin{proof}
($\Rightarrow$): Assume stabilizing consensus is solvable with a proper
decision function $\Delta$, and define $\Sigma_v = \Delta^{-1}(v)$. Since
$\Delta$ is semi-continuous by \cref{lem:Scont:Sconsensus}, it follows
from \cref{DefSemiCont:Lev} that $\Sigma_v$ is semi-open, as asserted in (1).
Validity immediately implies property~(2).

($\Leftarrow$): Assume every $\Sigma_v$ is semi-open and contains all $v$-valent
executions. We define a proper stabilizing consensus algorithm with weak
validity by defining the decision functions
$\Delta_p:\C\to\V$, $p\in\Pi$. Assume that the actual run is the admissible execution $\gamma=(C^t)_{t\geq 0}\in\Sigma$, 
which is of course unknown to the process $p \in \ob(\gamma)$, 
albeit $p$ can observe the sequence of local views $(V_p(C^{t_i}))_{t_i\geq 0}$, where $t_0=0$ and $t_i$
denotes the (unknown) global time of the $i$-th computing step of $p$.
Consider the set of executions defined by
{
\begin{align}
D_p(\gamma,{t_i}) &= \bigl\{\delta=(D^t)_{t\geq 0} \in \Sigma \mid (p \in \ob(\delta))\nonumber\\
&\qquad\qquad\qquad\qquad \wedge \exists s: C^{t_i} \sim_p D^s\bigr\} \nonumber\\
&= \bigl\{\delta=(D^t)_{t\geq 0} \in \Sigma \mid (p \in \ob(\delta)) \nonumber\\
&\qquad\qquad\qquad\qquad \wedge \exists s: V_p(C^{t_i})= V_p(D^s)\bigr\}.\nonumber
\end{align}
}
Note that $p$ can compute the set $D_p(\gamma,{t_i})$ at time ${t_i}$, since it
depends only on the current view $V_p(C^{t_i})$ and the sets $\Sigma_v$, $v\in\V$,
which we assume to be a priori known to $p$. Moreover, we will show below that 
$D_p(\gamma,{t_i})$ is open.

\begin{table*}
{
\begin{empheq}[left={\Delta_p(C^{t_i})=\empheqlbrace}]{align}
  v &\quad \mbox{if $D_p(\gamma,{t_i})\subseteq \Sigma_v$} \quad\mbox{(Case (a): $\gamma \in \inte(\Sigma_v)$ or $(\gamma \in \bdin\Sigma_v) \wedge (\gamma\not\in\bd\Sigma_w)$)},\label{eq:SCcase1} \\
  v &\quad \mbox{if, for all $w\neq v$, $D_p(\gamma,t_i)\cap \bdin\Sigma_w=\emptyset$} \quad\mbox{(Case (b): $\gamma \in \bdtwo{\bdin \Sigma_v}$ only)},\label{eq:SCcase2}\\
  v &\quad \mbox{if } D_p(\gamma,{t_i}) \cap \!\!\! \!\!\! \bigcap_{\substack{w \text{ where }\\
    D_p(\gamma,t_i)\cap \bdin \Sigma_w \neq \emptyset}}\!\!\! \!\!\!  \bdtwo{\bdin \Sigma_w} \subseteq \bdin \Sigma_v \quad\mbox{(Case (c): $\gamma \in \bdtwo{\bdin \Sigma_w}$ also)},\label{eq:SCcase3}\\
  w &\quad \mbox{otherwise, for any $w$ with $D_p(\gamma,{t_i}) \cap \Sigma_w \neq \emptyset$} \quad\mbox{(Case (d): $\gamma$ still unknown)},\label{eq:SCcase4}
\end{empheq}
}
\caption{Decision function for the proof of \cref{thm:Schar:nonunif}. Herein, $\bdtwo{\bdin \Sigma_w}$ denotes the boundary of the included boundary $\bdin \Sigma_w$ (taken in the subspace $Y=\bigcup_{v\in\V} \bdin \Sigma_v$).}
\label{tab:decfunc}
\end{table*}

We define $p$'s decision function for configuration $C^{t_i}$ at time ${t_i}$ in $\gamma$ as shown in \cref{tab:decfunc},
where we use $\bdtwo{\bdin \Sigma_w}$ to denote the boundary of the included boundary $\bdin \Sigma_w$ (taken in the subspace $Y=\bigcup_{v\in\V} \bdin \Sigma_v$).
The function~$\Delta_p$ is well defined, since the decision sets~$\Sigma_v$ and hence also $\bdin \Sigma_v$
are pairwise disjoint, $p\in\ob(\gamma)$, and $\gamma \in D_p(\gamma,{t_i})$.

To show Stabilizing Agreement (SA), we will distinguish several cases as illustrated in
\cref{fig:illustration}. Regarding $\gamma$ with $p\in\ob(\gamma)$, there are two
possibilities:

(1) $\gamma\in\inte(\Sigma_v)$ is an interior point. Then, 
there exists some $\ell\geq 0$ such that the open ball
$B_{2^{-\ell}}=\left\{ \delta\in \Sigma \mid \dnonunif(\gamma,\delta) < 2^{-\ell} \right\}
\subseteq
\Sigma_v$.
By definition of~$\dnonunif$, it follows that
$\dnonunif(\gamma,\delta) \leq d_p(\gamma,\delta)$ and hence
$B_{2^{-\ell}}^p(\gamma)=\bigl\{ \delta\in \Sigma \mid (p \in \ob(\delta)) \wedge d_p(\gamma,\delta) < 2^{-\ell} \bigr\}
\subseteq B_{2^{-\ell}}(\gamma) \subseteq \Sigma_v$.

Recalling that $\gamma = (C^t)_{t\geq0}$,
let~$L$ be the smallest integer such that $2^{-\chi_p(C^t)} \leq 2^{-\ell}$ for all $t\geq L$.
Such an~$L$ exists, since $\chi_p(C^t)\to\infty$ as $t\to\infty$.
Then, for every ${t_i}\geq L$, we find
{\small
\begin{align}
D_p(\gamma,{t_i}) &= \left\{ \delta\in \Sigma \mid (p \in \ob(\delta)) \wedge \exists s\colon C^{t_i} \sim_p D^s \right\}\nonumber\\
&\subseteq
\left\{ \delta\in \Sigma \mid (p \in \ob(\delta)) \wedge d_p(\gamma,\delta) < 2^{-\chi_p(C^{t_i})} \right\}\nonumber\\
&\subseteq B_{2^{-\ell}}^p(\gamma) \subseteq B_{2^{-\ell}}(\gamma) \subseteq \Sigma_v\label{eq:openincl}
\end{align}
}
Note that \cref{eq:openincl} also reveals that $D_p(\gamma,{t_i})$ is contained in an open set, hence open.
Anyway, it follows that $\Delta_p(C^{t_i}) = v$ for all ${t_i}\geq L$ according to \cref{eq:SCcase1} (Case (a): interior 
point, see \cref{fig:illustration}), 
i.e., process~$p$ eventually stabilizes on the decision value~$v$ in execution~$\gamma$.

To also show that all obedient processes decide the same value $v$ in $\gamma\in\inte(\Sigma_v)$, 
we assume for a contradiction
that process~$q \in \ob(\gamma)$ decides value $w\neq v$ in configuration $C^{t_i}$ in execution $\gamma$.
But then, by the definition of the function~$\Delta_q$, we must have
$\gamma \in D_q(\gamma,{t_i})=\bigl\{\delta=(D^t)_{t\geq 0} \in \Sigma \mid (q \in \ob(\delta)) \wedge \exists s: C^{t_i} \sim_p D^s\bigr\} \subseteq \Sigma_w$. 
But this is impossible, since $\gamma \in \Sigma_v$ and $\Sigma_v\cap\Sigma_w = \emptyset$.

\medskip

(2) $\gamma \in \bdin \Sigma_v$ is an included boundary point, and hence
a limit point of $\Sigma_v$ by \cref{def:properSC}. If $\gamma \not\in
\bd \Sigma_w$ for any $w\neq v$, then there will be some time ${t_i}$ from which
on $D_p(\gamma,{t_i})\subseteq \Sigma_v$ holds. The argument for case (1) above proves 
(SA) according to \cref{eq:SCcase1} (Case (a): boundary point, 
see \cref{fig:illustration}) in this case. Otherwise, $D_p(\gamma,{t_i})$ as a neighborhood of $\gamma$
intersects not ony $\Sigma_v$ (which is trivially the case since 
$\gamma \in D_p(\gamma,{t_i})$) but also in one or more other decision sets
$\Sigma_w$, since $\gamma$ is a limit point.

In order to determine the (a priori unknown) decision set $\Sigma_v$ containing $\gamma$, 
$\Delta_p$ considers the set $D_p(\gamma,{t_i})\cap \bdin \Sigma_w$ of candidate 
limit points in $\Sigma_w$, for all $w$ where this candidate set is non-empty.
We need to distinguish two subcases here: 

(b) If $\gamma$ is a limit point of the executions in
$\bdin \Sigma_v$, i.e., $\gamma \in  \bdtwo{\bdin \Sigma_v}$, but not also a 
limit point of any other $\bdin \Sigma_w$, i.e., $\gamma \not\in \bdtwo{\bdin \Sigma_w}$, 
there will be 
a time ${t_i}$ after which $D_p(\gamma,{t_i})\cap \bdin \Sigma_w=\emptyset$. Since
of course $\gamma \in D_p(\gamma,{t_i})\cap \bdin \Sigma_v$,
\cref{eq:SCcase2} (Case (b), see \cref{fig:illustration}) secures (SA), 
by invoking the analogous considerations involving $\chi_p$
that led to \cref{eq:openincl}.

(c) If there is at least one other candidate set for $w$, besides the one for $v$,
for which $\gamma$ is also a limit point of the limit executions in $\bdin \Sigma_w$, more specifically,
if there is a sequence of limit executions $\alpha_1,\alpha_2, \dots \in 
\bdin \Sigma_w$ with $\alpha_i \in B_{2^{-i}}(\alpha)$ for some limit
point $\alpha \in \bdtwo{\bdin \Sigma_w}$, in addition
to some sequence $\gamma_1,\gamma_2,\dots \in \bdin \Sigma_v$ with 
$\gamma_i \in  B_{2^{-i}}(\gamma)$ with a limit point $\gamma \in \bdtwo{\bdin \Sigma_v}$ satisfying
$\dnonunif(\alpha,\gamma)=0$, then the fact that $\gamma \in \Sigma_v$
also forces $\alpha \in \Sigma_v$ and hence $\alpha \in \bdin \Sigma_v$:
After all, the decision value 
of $\alpha$ and $\gamma$ must be the same since $\dnonunif(\alpha,\gamma)=0$.
Consequently, deciding $v$ according to \cref{eq:SCcase3} (Case (c), see 
\cref{fig:illustration}) ensures (SA) also in this case.

To show that all obedient processes decide the same value $v$ also here,
we again assume for a contradiction
that process~$q \in \ob(\gamma)$ decides value $w\neq v$ in configuration $C^{t_i}$ 
in execution $\gamma\in\Sigma_v$ by one of the cases in (2).
By the definition of the function~$\Delta_q$, we must have
$\gamma \in D_q(\gamma,{t_i})$. Since any of these cases requires 
$\gamma \in \Sigma_w$, this is impossible since $\gamma \in \Sigma_v$ 
and $\Sigma_v\cap\Sigma_w = \emptyset$.

If none of the above Cases (a)--(c) applies, $D_p(\gamma,{t_i})$ is still
too large to uniquely determine the decision set $\gamma$ belongs to. In
that case, $\Delta_p$ picks some possible decision value according to
\cref{eq:SCcase4} (Case (d), see \cref{fig:illustration}).

\medskip

Finally, Validity (V) immediately follows from property~(2). 
\end{proof}

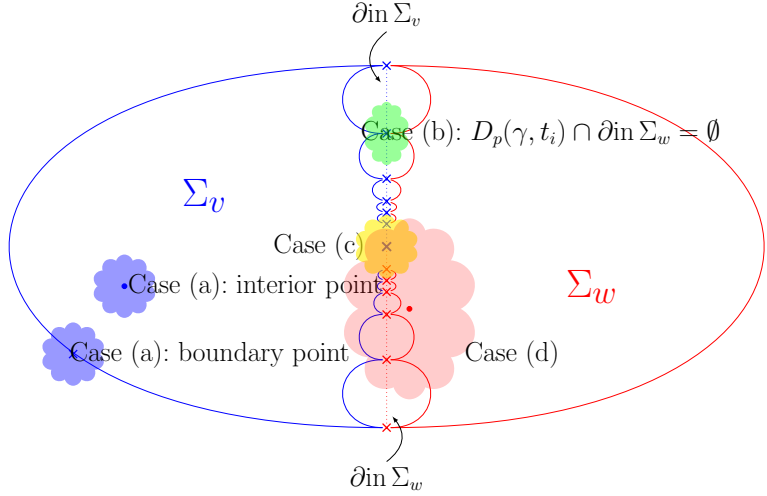
\begin{figure*}
\begin{center}
  \scalebox{0.6}{
\begin{tikzpicture}
	\begin{pgfonlayer}{nodelayer}
		\node [style=VBcross] (3) at (0, 7.75) {};
		\node [style=VBcross] (4) at (0, 4.75) {};
		\node [style=VBcross] (5) at (0, 2.75) {};
		\node [style=VBcross] (6) at (0, 1.75) {};
		\node [style=VBcross] (7) at (0, 1.25) {};
		\node [style=VBcross, very thick] (8) at (0, -0.25) {};
		\node [style=VRcross] (9) at (0, -1.75) {};
		\node [style=VRcross] (10) at (0, -2.25) {};
		\node [style=VRcross] (11) at (0, -3.25) {};
		\node [style=VRcross] (12) at (0, -8.25) {};
		\node [style=VRcross] (13) at (0, -5.25) {};
		\node [style=VBcross] (14) at (0, 0.75) {};
		\node [style=VRcross] (15) at (0, -1.25) {};
		\node [style=none, color=blue] (16) at (-8, 2) {{\Huge $\Sigma_v$}};
		\node [style=none, color=red] (17) at (9, -2) {{\Huge $\Sigma_w$}};
		\node [style=none] (21) at (0, 9.25) {};
		\node [style=none] (25) at (0, -9.75) {};
		\node [style=none] (26) at (-0.25, 7) {};
		\node [style=none] (27) at (0.25, -7.75) {};
		\node [style=none] (30) at (2, 4.75) {};
		\node [style=none] (31) at (6.75, 4.75) {{\Large Case (b): $D_p(\gamma,t_i)\cap \bdin \Sigma_w = \emptyset$}};
		\node [style=none] (32) at (-2.95, -0.25) {{\Large Case (c)}};
		\node [style=none] (33) at (0, 10) {{\Large $\bdin \Sigma_v$}};
		\node [style=none] (34) at (0, -10.5) {{\Large $\bdin \Sigma_w$}};
		\node [style=VYfadedcircle] (35) at (0, -0.25) {};
		\node [style=VGfadedcircle] (36) at (0, 4.75) {};
		\node [style=VBpoint] (37) at (-11.5, -2) {};
		\node [style=VBfadedcircle] (38) at (-11.5, -2) {};
		\node [style=none] (39) at (-5.75, -2) {{\Large Case (a): interior point}};
		\node [style=VBcross] (40) at (-13.75, -5) {};
		\node [style=VBfadedcircle] (41) at (-13.75, -5) {};
		\node [style=none] (42) at (-7.75, -5) {{\Large Case (a): boundary point}};
		\node [style=VRpoint] (43) at (1, -3) {};
		\node [style=VRfadedcircle] (44) at (1, -3) {};
		\node [style=none] (45) at (5.5, -5) {{\Large Case (d)}};
	\end{pgfonlayer}
	\begin{pgfonlayer}{edgelayer}
		\draw [style=EBnone, bend right=90, looseness=3.50] (3) to (12);
		\draw [style=ERnone, bend left=90, looseness=3.50] (3) to (12);
		\draw [style=EBnone, bend right=90, looseness=2.00] (3) to (4);
		\draw [style=EBnone, bend right=90, looseness=1.75] (4) to (5);
		\draw [style=EBnone, bend right=75, looseness=1.75] (5) to (6);
		\draw [style=EBnone, bend right=75, looseness=2.25] (6) to (7);
		\draw [style=ERnone, bend left=75, looseness=2.00] (9) to (10);
		\draw [style=ERnone, bend left=75, looseness=1.75] (10) to (11);
		\draw [style=ERnone, bend left=90, looseness=1.75] (11) to (13);
		\draw [style=ERnone, bend left=90, looseness=2.00] (13) to (12);
		\draw [style=ERnone, bend left=90, looseness=2.00] (3) to (4);
		\draw [style=ERnone, bend left=90, looseness=1.75] (4) to (5);
		\draw [style=ERnone, bend left=75, looseness=1.75] (5) to (6);
		\draw [style=ERnone, bend left=75, looseness=2.25] (6) to (7);
		\draw [style=EBnone, bend right=75, looseness=2.00] (9) to (10);
		\draw [style=EBnone, bend right=75, looseness=1.50] (10) to (11);
		\draw [style=EBnone, bend left=270, looseness=1.75] (11) to (13);
		\draw [style=EBnone, bend right=90, looseness=2.00] (13) to (12);
		\draw [style=EBnone, bend right=75, looseness=2.00] (7) to (14);
		\draw [style=EBnone, bend right=75, looseness=1.75] (15) to (9);
		\draw [style=ERnone, bend left=75, looseness=2.00] (7) to (14);
		\draw [style=ERnone, bend left=75, looseness=2.00] (15) to (9);
		\draw [style=EBnonedotted] (3) to (8);
		\draw [style=ERnonedotted] (8) to (12);
		\draw [style=Eout, bend right=45, looseness=1.25] (21.center) to (26.center);
		\draw [style=Eout, bend right=45, looseness=1.25] (25.center) to (27.center);
	\end{pgfonlayer}
\end{tikzpicture}
  }
\end{center}
\caption{Illustration of the two cases in the proof of \cref{thm:Schar:nonunif}, depicting the four different cases for the decision function $\Delta_p(C^{t_i})$. The blue, green, yellow, and red cloud represents $D_p(\gamma,t^i)$ 
for Case (a), (b), (c), and (d), respectively. 
The included boundaries $\bdin \Sigma_v$ resp.\ $\bdin \Sigma_w$ are formed by the limit points (marked by a
cross) that lie 
on the blue resp.\ red dotted lines. The thick limit
point at the center of the yellow cloud lies in the intersection of $\bdtwo{\bdin \Sigma_v}$ and $\bdtwo{\bdin \Sigma_w}$; we
assume here w.l.o.g. that it belongs to $\bdin \Sigma_v$ (blue).}
\label{fig:illustration}
\end{figure*}

\medskip

We conclude this section with \cref{lem:nonproperreduction}, which shows that non-proper stabilizing 
consensus algorithms can indeed be turned into proper ones: 

\begin{theorem}[Reduction of non-proper to proper]\label{lem:nonproperreduction} 
Any decision function $\Delta'$ of a non-proper stabilizing consensus algorithm that 
satisfies \cref{def:properSC}.(i) can be turned into a decision function $\Delta$ 
that also satisfies \cref{def:properSC}.(ii).
\end{theorem}
\begin{proof}
To construct $\Delta$, we set $\Delta(\gamma)=\Delta'(\gamma)$ for all executions
$\gamma$ not violating (ii). For every $\gamma$ that violates (ii), the following
change is applied: Since $\gamma\in \bdin \Sigma_v$ is not an isolated point
in $\Sigma$, it must satisfy $\gamma \in \bd \Sigma_w$ for at least one $w\neq v$. 
We pick any such $w$ and define $\Delta(\gamma)=w$, i.e., we change the original
decision value from $v$ to $w$. Note that this amounts to moving $\gamma$ from
$\bdin \Sigma_v$ to $\bdin \Sigma_w$.

Whereas this change cannot invalidate stabilizing agreement (SA) provided by
the original $\Delta'$, we need to make sure that changing the decision value
of $\gamma$ from $v$ to $w$ does not violate validity (V). This is impossible,
however, as this could only happen if $\gamma$ was $v$-valent, in which
case $\gamma \in \bd \Sigma_w$ could not have occured.
\end{proof}

\section{Equivalence of Stabilizing Consensus with Weak and Strong Validity}
\label{sec:broadcastability}

A byproduct of the topological characterization for terminating consensus developed
in \cite{NSW24:JACM} is the equivalence of weak validity (V) and strong validity (SV) (recall
\cref{def:consensus}). For binary consensus, i.e., $|\V|=2$, this is a well-known 
fact \cite[Ex.~5.1]{AW04}, for larger input sets, it was, to the best of 
our knowledge, not known before. In this section, we will show that this is also
true for stabilizing consensus. 

We start with some definitions needed for formalizing this condition:

\begin{definition}[Heard-of sets {\cite[Def.~7.1]{NSW24:JACM}}]\label{def:HO}
For every process~$p\in\Pi$, there is a function $\ho_p:\C\to2^\Pi$ that maps a configuration~$C\in\C$ to the set of processes~$\ho_p(C)$ that~$p$ has (transitively) heard of in~$C$. Its extension to execution $\gamma=(C^t)_{t\geq 0}$ 
is defined as $\ho_p(\gamma) = \bigcup_{t\geq0} \ho_p(C^t)$.

Heard-of sets have the following obvious properties: For executions $\gamma = (C^t)_{t\geq0}$, $\delta = (D^t)_{t\geq0}$ and all $t\geq 0$,
\begin{enumerate}
\item[(i)] $p\in\ho_p(C^t)$, and $\ho_p(C^t) = \ho_p(D^t)$ if $C^t \sim_p D^t$,
\item[(ii)] $\ho_p(C^t)\subseteq \ho_p(C^{t+1})$,
\item[(iii)] for all $x\in\Pi$, if $x\in \ho_q(C^t) \cap \ho_q(D^t)$ and $C^t \sim_q D^t$, then $I_x(\gamma) = I_x(\delta)$ (where~$I_p(\gamma)$ denotes the initial value of process~$p$ in execution~$\gamma$).
\end{enumerate}
\end{definition}

The independent arbitrary input assignment condition stated in \cref{def:independentinputassumption} secures that, for every execution $\gamma$ with initial value assignment $I(\gamma)$, there is a
an isomorphic execution $\delta$ w.r.t.\ the HO sets of all processes 
that starts from an arbitrary other initial value assignment $I(\delta)$.

\begin{definition}[Independent arbitrary input assignment condition {\cite[Def.~7.2]{NSW24:JACM}}]\label{def:independentinputassumption}
Let $I:\Pi\to\V$ be some assignment of initial values to the processes, and
$\Sigma^{(I)}\subseteq \Sigma$ be the set of admissible executions with that initial value assignment. We say that $\Sigma$ satisfies the \emph{independent input
assignment condition}, if and only if 
for any two assignments~$I$ and~$J$, we have $\Sigma^{(I)} \cong \Sigma^{(J)}$,
that is, there is a bijective mapping $f_{I,J}:\Sigma^{(I)}\to\Sigma^{(J)}$ such that
for all $\gamma = (C^t)_{t\geq 0} \in \Sigma^{(I)}$ and $\delta = (D^t)_{t\geq 0} \in \Sigma^{(I)}$, writing $f_{I,J}(\gamma) = (C_f^t)_{t\geq 0}$ and $f_{I,J}(\delta) = (D_f^t)_{t\geq 0}$, the following holds for all $t\geq 0$ and all $p\in\Pi$:
\begin{enumerate}
\item $\ob(C^t) = \ob(C_f^t)$ 
\item $C^t \sim_p D^t$ if and only if $C_f^t \sim_p D_f^t$
\item $\ho_p(C^t) = \ho_p(C_f^t)$
\item $C^t \sim_p C_f^t$ if $I_q = J_q$ for all $q\in \ho_p(C^t)$
\end{enumerate}
We say that $\Sigma$ satisfies the \emph{independent arbitrary input
assignment condition}, if it satisfies the independent input assignment
condition for every choice of $I:\Pi\to\V$.
\end{definition}

\cref{lem:broadcastableCCs} below will reveal that if proper stabilizing 
consensus (with weak validity) and independent arbitrary inputs 
is solvable, then every connected component in~$\Sigma_v$ needs 
to be broadcastable.

\begin{definition}[Broadcastability {\cite[Def.~7.3]{NSW24:JACM}}]\label{def:broadcastability}
We call a subset $A \subseteq \Sigma$ of admissible executions \emph{broadcastable} 
by the broadcaster $p\in \Pi$, if, in every execution $\gamma \in A$, every obedient process $q\in\ob(\gamma)$ eventually hears from process~$p$, i.e., $p\in\ho_q(\gamma)$, and hence knows $I_p(\gamma)$.
\end{definition}

\begin{lemma}[Broadcastable connected components]\label{lem:broadcastableCCs}
A connected component $\Sigma_\gamma$, containing $\gamma$, of any decision set
$\Sigma_v$ for proper stabilizing consensus (i.e., with independent arbitrary input 
assignments) that is not broadcastable for some process would contain $w$-valent 
executions, for every $w\in \V$.
\end{lemma}

\begin{proof}
To prove our lemma,
we consider the finite sequence of executions $\gamma=\alpha_0,\alpha_1,\dots,\alpha_n=\gamma_w$
obtained from $\gamma$ by changing the initial values of the processes $1,\dots,n$ in $I(\gamma)$ to an arbitrary but fixed $w$, one by one (it is here where we need the arbitrary
input assignment assumption). We show by induction that $\alpha_p \in \Sigma_{\gamma}$ for every $p\in \{0,\dots,n\}$, which proves our claim since $\alpha_n=\gamma_w$ is $w$-valent. 

The induction basis $p=0$ is trivial, so suppose $\alpha_{p-1}\in \Sigma_{\gamma}$ according to the
induction hypothesis. If it happens that $I_p(\alpha_{p-1})=I_{p}(\gamma)=w$ already, nothing needs to be done and we just set $\alpha_{p}=\alpha_{p-1} \in \Sigma_{\gamma}$. Otherwise, $\alpha_{p}$ is $\alpha_{p-1}$ with the initial value $I_{p}(\alpha_{p})$ changed to $w$. Now suppose for a contradiction that $\alpha_{p} \in\Sigma_{\alpha_{p}} \neq \Sigma_{\gamma}$. 

Since $\Sigma_{\gamma}$ is not broadcastable by any process, hence also not by $p$, there is some
execution $\eta\in \Sigma_{\gamma}$ with $\eta = (C^t)_{t\geq0}$ and a process $q\neq p$ with $q \in \ob(C^t)$ and the initial value $I_p(\eta)$ not in $q$'s view $V_{q}(C^t)$ for every $t\geq 0$.
Thanks to the independent input assignment property \cref{def:independentinputassumption}, there is 
also an execution $\delta=f_{I(\eta),I'}(\eta) \in \Sigma_{\alpha_{p}}$ that
matches $\eta$, i.e., is the same as $\eta$ except that $I(\delta)=I'$ with $I'_q=I_q(\eta)$ for $p \neq q \in \Pi$ but possibly $I'_p\neq I_p(\eta)$.
It follows that $d_{q}(\eta,\delta)=0$ with $q\in\ob(\eta)\cap\ob(\delta)$ and hence  $\dnonunif(\eta,\delta)=0$. Consequently, $\delta \in \Sigma_{\gamma}$
and hence $\Sigma_{\alpha_{p}}=\Sigma_{\gamma}$, which provides the required contradiction and
completes the induction step.
\end{proof}

As a consequence of \cref{lem:broadcastableCCs}, it turns out that \emph{any} connected broadcastable set has 
a diameter strictly smaller than~$1$ in our non-uniform topology.

\begin{definition}[Diameter of a set]\label{def:diameterset}
For $A \subseteq \Comega$, depending on the distance function $d$ that induces
the appropriate topology,
define $A$'s diameter as $d(A)=\sup\{d(\gamma,\delta)\mid \mbox{$\gamma, \delta\in A$}\}$.
\end{definition}

\begin{lemma}[Diameter of broadcastable connected sets {\cite[Lem.~7.6]{NSW24:JACM}}]\label{lem:broadcastablediameter}
If a connected set $A\subseteq \Sigma$ 
of admissible executions is broadcastable by some process $p$, then
$\dunif(A) \leq d_{p}(A)\leq 1/2$, as well as $\dnonunif(A) \leq 1/2$, i.e., $p$'s initial value satisfies $I_p(\gamma)=I_p(\delta)$ for all $\gamma,\delta \in A$. 
\end{lemma}
\begin{proof}
  Our proof below for $d_p(A)\leq 1/2$ translates literally to any $d \in \{d_p, \dunif, \dnonunif\}$;
  the statement $\dunif(A) \leq d_{p}(A)$ follows from
the definition in \cref{eq:dunif}.

Broadcastability by $p$ implies that, for any $\gamma\in A$ with $\gamma=(C^t)_{t\geq0}$, every process $q$ has $I_p(\gamma)$ in its local view 
$V_{q}(C^{T(\gamma)})$ for some $0<T(\gamma)<\infty$ or is not obedient any more. 
Abbreviating $t=T(\gamma)$, consider any $\delta\in B_{2^{-t}}(\gamma) \cap A$
with $\delta = (D^t)_{t\geq0}$.
By definition of $B_{2^{-t}}(\gamma)$, there must be some process $q \in
\ob(D^t)\cap\ob(C^t)$ with $V_{q}(D^{t})=V_{q}(C^{t})$. \cref{def:HO}.(iii)
thus guarantees $I_p(\delta)=I_p(\gamma)$. 

We show
now that this argument can be continued to reach every $\delta\in A$. For a contradiction, suppose that this is not the case and let $U(\gamma)$ be the union of the balls
recursively defined as follows: $U_0(\gamma)=\{\gamma\}$, for $m>0$,
$U_m(\gamma) = \bigcup_{\delta \in U_{m-1}(\gamma)} (B_{2^{-T(\delta)}}(\delta)
\cap A)$,  and finally $U(\gamma)=\bigcup_{m\geq 0} U_m(\gamma)$.
As a union of open balls intersected with $A$, which are all open in $A$, both
$U_m(\gamma)$ for every $m > 0$ and $U(\gamma)$ is hence open in $A$.
For every $\delta \in A\setminus U(\gamma)$, $U(\delta)$ is also open in~$A$, and so is $V(\gamma)=\bigcup_{\delta \in A\setminus U(\gamma)}U(\delta)$. However, the open sets $U(\gamma)$ and $V(\gamma)$ must satisfy $U(\gamma) \cap V(\gamma) = \emptyset$ (as
they would be the same otherwise) and $U(\gamma)\cup V(\gamma)=A$, hence $A$ cannot be connected. 
\end{proof}

\cref{lem:broadcastableCCs} in conjunction with \cref{lem:broadcastablediameter} finally implies:

\begin{corollary}[Broadcastable $\Sigma_{\gamma}$]\label{cor:broadcastablePS}
If proper stabilizing consensus with weak validity is solvable, 
then every connected component $\Sigma_{\gamma}\subseteq \Sigma_v$ must be broadcastable by some process $p$. In every execution $\gamma'\in\Sigma_{\gamma}$, the broadcaster $p$ has the same initial value $I_p(\gamma')$ (and, of course, the same decision value $v$).
\end{corollary}

To emphasize the key role of \cref{cor:broadcastablePS} for the equivalence of weak validity (V) 
and strong validity (SV), note that
the transition from (V) to (SV) in \cref{thm:Schar:nonunif} just requires the replacement 
of condition (2), i.e., \emph{``If execution $\gamma\in\Sigma$ is $v$-valent, then $\gamma \in \Sigma_v$''}, by 
\emph{``If execution $\gamma\in\Sigma_v$, then there is a process $p$ with initial
value $I_p(\gamma)=v$''.} This change results in a strong version
of our theorem, since this modification is transparent
for the proof of \cref{thm:Schar:nonunif}. Note also
that both versions are equivalent for $v$-valent executions.

The crucial role of \cref{cor:broadcastablePS} is that it \emph{always} allows to turn a proper weak stabilizing
consensus algorithm into a strong one, as it reveals that if weak stabilizing consensus is solvable, then
every connected component $\Sigma_\gamma \subseteq \Sigma_v$ must have at least one common broadcaster $b=b(\gamma')=b(\Sigma_\gamma)$ that has the same initial value $I_b(\gamma')=I_b(\gamma)= I_b(\Sigma_\gamma)$ (which may, however, be different
from $v$) in all executions $\gamma'\in\Sigma_\gamma$. Consequently, if decision sets 
exist that allow to solve consensus with weak validity according to
\cref{thm:Schar:nonunif}, one can always
reshuffle the connected components among the decision sets to form \emph{strong} 
decision sets, which use the initial value of some broadcaster
for assigning a connected component to a decision set:

\begin{definition}[Strong decision sets]\label{def:strong} Let $\Sigma$ be the
set of admissible executions of any (weak or strong) proper stabilizing consensus algorithm 
with independent arbitrary input assignments. A \emph{strong decision set} $\Sigma_v$ for $v\in \V$
satisfies
\begin{equation}
\Sigma_v = \bigcup_{p \in \Pi}  \Sigma_v^p  \qquad\mbox{with}\qquad \Sigma_v^p = \bigcup_{\substack{\gamma\in\Sigma\\ b(\Sigma_\gamma)=p \\ I_{p}(\gamma)=v}} \Sigma_\gamma 
\label{eq:Sigmavp}.
\end{equation}
\end{definition}
Note that strong decision sets need not be unique, as some connected component
$\Sigma_\gamma$ might have several broadcasters, any of which could be used for
determining its decision value $v$. The canonical choice to make it uniquely
defined would be to take the lexically smallest $p=b(\Sigma_\gamma)$ among
all broadcasters $p' \geq p$ in $\Sigma_\gamma$. 

Practically, all that needs to be done for a connected component $\Sigma_\gamma \subseteq \Sigma_v$
that only has a broadcaster $p$ with initial value $I_p(\gamma)=w\neq v$ in every execution $\gamma$
is to change the decision value from $v$ to $w$, which will automatically move $\Sigma_\gamma$ from
$\Sigma_v$ to $\Sigma_w$. It must be noted, though, that the actual utility of our result is limited:
Since our solution algorithm depends on the a priori knowledge of the decision sets, it does
not give any clue on how to develop a practical strong consensus algorithm from a practical
weak consensus algorithm in a given model. In fact, determining and agreeing upon a broadcaster 
in executions that are not $v$-valent is a very hard problem.

Anyway, along with the obvious fact that strong validity implies weak validity,
our findings reveal:

\begin{corollary}[Equivalence of weak and strong validity]\label{cor:weakisstrong} 
Proper stabilizing consensus with weak validity is solvable in a model if and only 
if consensus with strong validity is solvable in this model.
\end{corollary}

\section{Applications}
\label{sec:applications}

In this section, we will apply our findings to the few existing results
on stabilizing consensus. This way, we will provide a topological
explanation of why the problem is solvable/impossible in certain models.

\subsection{Stabilizing consensus in asynchronous crash-prone systems}
\label{sec:asyncwithcrashes}

We start with the simple stabilizing consensus algorithm for
crash-prone asynchronous systems proposed in \cite{AFJ06}.
Designed for any number $f < n$ of crashes, and processes that
are connected by fair-lossy point-to-point links, it works as
follows: In every step, every process $p$ broadcasts its current
decision value $O_p$. If a message containing value $v$ has been
received from any process, $p$ sets $O_p = \min\{O_p,v\}$.
Note that the processes execute their algorithm without
any synchronization, which implies that we can just take 
$\chi_p(t)=t$, i.e., global real-time, in our analysis.

The above algorithm gives raise to (strong) decision sets $\Sigma_v$, $v\in \V$,
defined by $\Sigma_v = \bigl\{\gamma \in \Sigma\mid \min_{p \in bc(\gamma)}\{I_p(\gamma)\}=v\bigr\}$,
where $bc(\gamma)=\bigl\{p\mid p \in \bigcap_{q \in \ob(\gamma)}\ho_q(\gamma)\bigr\}$ denotes 
the set of broadcasters in $\gamma$. Note that the processes in $bc(\gamma)$ 
need not necessarily be obedient.

To show that every $\Sigma_v$ is semi-open, we need to verify that every 
$\gamma \in \Sigma_v$ is either (i) an interior point ($\gamma \in \inte(\Sigma_v)$) 
or else (ii) a point in the included boundary ($\gamma \in \bdin \Sigma_v$).
Writing $\gamma=(C^t)_{t\geq0}$, $\gamma\in\Sigma_v$ implies that there is
some $p\in bc(\gamma)$ with $I_p(\gamma)=v$ and a time $t$ such that $p \in \ho_q(C^t)$ for 
every $q\in\ob(\gamma)$. Fix some $t$ where this holds for every process
$p\in bc(\gamma)$ with $I_p(\gamma)=v$ (there may be several), and assume that there is some 
$\delta=(D^t)_{t\geq 0} \in B_{2^{-t}}(\gamma) \cap \Sigma_w$, where we first assume $w > v$. By the
definition of $\dnonunif$ (\cref{eq:dnonunif}),
there must hence be some process $s \in \ob(\gamma)\cap\ob(\delta)$ that has the same view
$V_s(D^t)=V_s(C^t)$ at time $t$. This implies, however, that $p\in\ho_s(\delta)$, so $s$ also knows $v$. 
Since $s\in\ob(\delta)$,
$s$ will eventually broadcast $v$ to all obedient processes in $\delta$ (recall the fair-lossy
link assumption), which reveals that
$p \in bc(\delta)$ as well. Consequently, no such $\delta$ can exist, which results in
$B_{2^{-t}}(\gamma) \subseteq \Sigma_v$ and thus $\gamma \in \inte(\Sigma_v)$ according to (i).

However, in the argument above, we might observe $\delta \in B_{2^{-t}}(\gamma) \cap \Sigma_w$ for
some $w < v$. In this case, the existence of $\delta$ would not be contradictory as before. 
Actually, there are two possibilities here: If there is some finite $t' > t$ such that $B_{2^{-t'}}(\gamma) \cap \Sigma_w=\emptyset$,
we find $B_{2^{-t'}}(\gamma) \subseteq \Sigma_v$, which confirms $\gamma \in \inte{\Sigma_v}$.
If, however, an ``offending'' $\delta_i' \in B_{2^{-i}}(\gamma) \cap \Sigma_w$ 
can be found for every $i\geq t$, then $\gamma = \lim_{i\to\infty} 
\delta_i' \in \bd \Sigma_w$ is a limit point of $\Sigma_w$. Obviously, none of the broadcasters $p' \in bc(\delta_i')$ with
$I_{p'}(\delta_i')=w<v$ can also satisfy $p' \in bc(\gamma)$. 
The only way how a sequence of $\delta_1',\delta_2',\dots \in \Sigma_w$ can converge to $\gamma \in \Sigma_v$ is hence when
\emph{all} the obedient processes in $\delta_i'$ hear of the value $w$ broadcast by
$p'$ not before or at time $i$: Since $\delta_i' \in B_{2^{-i}}(\gamma)$,
no process that is obedient in both $\gamma$ and $\delta_i'$ can have heared from $p'$. This, in turn, 
is only possible if $p'$ is not obedient in $\gamma$, in a way that it never tells any obedient process about $w$.
But this implies that the sequence $\delta_1,\delta_2, \dots$, which is the same as $\delta_1',\delta_2',\dots$
except that also $p'$ has the initial value $I_{p'}(\delta_i)=v$, satisfies $\delta_i \in \Sigma_v$ and
converges to $\gamma$ as well. Consequently, $\gamma \in \bdin \Sigma_v$ is a limit point according to case (ii).

Therefore, all $\Sigma_v$ are indeed strong and semi-open.

\subsection{Stabilizing consensus in synchronous systems controlled by a message adversary}
\label{sec:syncMA}

We next turn our attention to synchronous dynamic networks controlled 
by a message adversary \cite{AG13}. We will provide topological
explanations of the (few) core results on stabilizing consensus
established in the past \cite{CM19:DC,SS21:SSS,FR24:arxiv}. Most
of these results are from \cite{CM19:DC} or are based on the MinMax algorithm
introduced there.

The general setting is a lock-step synchronous system of $n$ fault-free
processes $\Pi$, so $\ob(\gamma)=\Pi$ for every $\gamma=(C^t)_{t\geq 0} \in \Sigma$
and $\chi_p(C^t)=t$ for all $p \in \Pi$. The communication in every round
is controlled by a message adversary, which determines the sequence of communication 
graphs $\G=(\G^t)_{t\geq 1}$, subsequently called \emph{communication pattern}, that 
governs execution $\gamma$. A message adversary can be identified
with the set of admissible communication patterns it allows.  Since all our algorithms
are deterministic, $\G$ and the input assignment $I(\gamma)$ uniquely determines 
$\gamma$. Note that we
implicitly assume that every communication graph always contains all
self-loops $p\rightarrow p$ for $p \in \Pi$.

\subsubsection{Lossy-link model possibility}
\label{sec:LLmodel}

Due to its importance for the analysis of the DLL model \cite{FR24:arxiv}
considered in \cref{sec:DLLmodel}, we will start with stabilizing consensus in the
lossy-link model \cite{SW89,SWK09,CGP15}. The underlying system
consists of only two processes $\Pi=\{l,r\}$ (the ``left'' and the 
``right'' process), where the communication graph in every round
of the communication pattern $\G=(\G^t)_{t\geq 1}$ of an execution $\gamma=(C^t)_{t\geq 0}$ 
is taken from the set $\LL=\{l \leftarrow r, l \leftrightarrow r,  
l \rightarrow r\}$. For conciseness, we will abbreviate $\LL = \{\leftarrow,
\leftrightarrow, \rightarrow\}$, with the implicit meaning $\leftarrow \; = \; l \leftarrow r$ etc., and write $\G \in \LLomega$ in the sequel.

The LL model is the most prominent example of an \emph{oblivious} message
adversary \cite{CGP15}. It is well-known, see e.g.\ \cite{SW89,SWK09,CGP15},
that terminating consensus is impossible here. Stabilizing consensus can
be solved in the LL model, however, by using a simple MinMax algorithm:
Translating the notation from \cite{FR24:arxiv} to our setting, the
decision function of process $p$ (computed at the end of round $t\geq 1$ in $\gamma=(C^t)_{t\geq 0}$) can be expressed as
\begin{equation}
\Delta_p(C^t) = \max_{q:V_q(C^{t-1}) \in V_p(C^t)} \bigl\{ \min_{s \in \ho_q(C^{t-1})}I_s(\gamma) \bigr\}\label{def:MinMax}.
\end{equation}
To interprete this definition, note that $p$'s view $V_p(C^t)$ is the (in the case of
anonymous processes, disjoint) union of the views $V_q(C^{t-1})$ that process $p$
receives from the processes (including $q=p$) in round $t$.

It is easy to see that this $\Delta_p$ solves stabilizing consensus in every $\gamma \in \Sigma$. Indeed,
the only non-trivial case is when $I_l(\gamma)\neq I_r(\gamma)$: Assuming w.l.o.g.\ 
$m=I_l(\gamma) < I_r(\gamma)=m'$, if $r$ receives a message from $l$ in some round 
$t < \infty$ for the first time, it sets $O_r=m$ in round $t+1$. Moreover, 
$O_l=m$ is set by process $r$ in round $t$ if it does not receive a message
from $l$ in round $t$, and in round $t+1$ otherwise, irrespectively of
whether it gets a message from $r$ or not. If, on the other hand, $r$ never receives such a message, 
$\ho_r(C^t)=\{r\}$ for all $t\geq 0$. Hence, $O_r=m'$ from $t=0$ on, and since
the maximum view found in $\Delta_l(C^t)$ is the one of $q=r$, $O_l=m'$ from
round $t=1$ on as well.

As in \cref{sec:asyncwithcrashes}, this algorithm gives raise to strong decision
sets $\Sigma_v$, $v\in \V$, defined by $\Sigma_v = \bigl\{\gamma \in \Sigma\mid 
\min_{p \in bc(\gamma)}\{I_p(\gamma)\}=v\bigr\}$,
where $bc(\gamma)=\bigl\{p\mid p \in \bigcap_{q \in \Pi}\ho_q(\gamma)\bigr\}$ denotes 
the set of broadcasters in $\gamma$. 
The proof is essentially the same as in \cref{sec:asyncwithcrashes}, except that $\delta_i'\in \Sigma_w$ in the sequence 
$\delta_1',\delta_2', \dots \to \gamma \in \Sigma_v$ lets the first message conveying the smaller 
input value $w < v$ to the other process arrive at time $i$. Note that the only boundary points can be
executions $\gamma$ where either $I_l(\gamma) < I_r(\gamma)$ and $\G = \{\leftarrow\}^\omega$ or else
$I_l(\gamma) > I_r(\gamma)$ and $\G =  \{\rightarrow\}^\omega$ here, since otherwise the
broadcaster with initial value $w<v$ would be a broadcaster in $\gamma \in \Sigma_v$, which is
impossible.

Obviously, the LL model considered above is only suitable if both processes
are active from $t=0$ on. To also cover $t_r^a >0$ and/or $t_l^a > 0$ allowed
by the system model of \cref{sec:general:model}, 
$\diamond LL = \{\noarrow\}^*LL^{\omega}$, which allows an arbitrary finite
prefix of empty graphs before the LL suffix, is the appropriate model. It is
easy to see that the MinMax algorithm given above, as well as our
topological characterization, also applies unchanged for $\diamond LL$.

\medskip

In order to prepare the grounds for dealing with the DLL model in 
\cref{sec:DLLmodel}, we briefly introduce some additional facts
about the LL model for $\V=\{0,1\}$, see \cite{SZ00:SIAM,CGP15} for details.
Most importantly, it is possible to totally order the $k$-prefixes of
all admissible graph sequences such that consecutive prefixes are
indistinguishable for one of the processes in $\Pi=\{l,r\}$, provided
all executions start with the same input assignment. 
\cref{fig:prefixorder} illustrates this for $k=1$ and $k=2$, where the white resp.\ 
black nodes represent process $l$ resp.\ $r$.

\begin{figure}[ht]
\begin{center}
  \scalebox{0.8}{
\begin{tikzpicture}
	\begin{pgfonlayer}{nodelayer}
		\node [style=VS] (7) at (-5, -2) {};
		\node [style=VW] (13) at (1, -2) {};
		\node [style=VW] (14) at (-11, -2) {};
		\node [style=VS] (15) at (7, -2) {};
		\node [style=none] (22) at (-8, -1) {\footnotesize $\Gamma_{\rightarrow}$};
		\node [style=none] (23) at (-2, -1) {\footnotesize $\Gamma_{\leftrightarrow}$};
		\node [style=none] (24) at (4.5, -1) {\footnotesize $\Gamma_{\leftarrow}$};
		\node [style=VSsmall] (25) at (-9, -7) {};
		\node [style=VSsmall] (26) at (-5, -7) {};
		\node [style=VSsmall] (27) at (-1, -7) {};
		\node [style=VSsmall] (28) at (3, -7) {};
		\node [style=VSsmall] (29) at (7, -7) {};
		\node [style=VWsmall] (30) at (-11, -7) {};
		\node [style=VWsmall] (31) at (-7, -7) {};
		\node [style=VWsmall] (32) at (-3, -7) {};
		\node [style=VWsmall] (33) at (1, -7) {};
		\node [style=VWsmall] (34) at (5, -7) {};
		\node [style=none] (45) at (-10, -6) {\footnotesize $\Gamma_{\rightarrow\rightarrow}$};
		\node [style=none] (46) at (-8, -6) {\footnotesize $\Gamma_{\rightarrow\leftrightarrow}$};
		\node [style=none] (47) at (-6, -6) {\footnotesize $\Gamma_{\rightarrow\leftarrow}$};
		\node [style=none] (48) at (0, -6) {\footnotesize $\Gamma_{\leftrightarrow\rightarrow}$};
		\node [style=none] (49) at (-2, -6) {\footnotesize $\Gamma_{\leftrightarrow\leftrightarrow}$};
		\node [style=none] (50) at (-4, -6) {\footnotesize $\Gamma_{\leftrightarrow\leftarrow}$};
		\node [style=none] (51) at (6, -6) {\footnotesize $\Gamma_{\leftarrow\leftarrow}$};
		\node [style=none] (52) at (4, -6) {\footnotesize $\Gamma_{\leftarrow\leftrightarrow}$};
		\node [style=none] (53) at (2, -6) {\footnotesize $\Gamma_{\leftarrow\rightarrow}$};
	\end{pgfonlayer}
	\begin{pgfonlayer}{edgelayer}
		\draw [style=Einout] (7) to (13);
		\draw [style=Eout, in=180, out=0] (14) to (7);
		\draw [style=Ein] (13) to (15);
		\draw [style=Eout] (30) to (25);
		\draw [style=Eout] (33) to (28);
		\draw [style=Ein] (31) to (26);
		\draw [style=Ein] (34) to (29);
		\draw [style=Einout] (25) to (31);
		\draw [style=Einout] (32) to (27);
		\draw [style=Einout] (28) to (34);
		\draw [style=Eout] (26) to (32);
		\draw [style=Ein] (27) to (33);
	\end{pgfonlayer}
\end{tikzpicture}
  }
\end{center}
\caption{Prefix order for the 1-prefixes and 2-prefixes in the LL model.}
\label{fig:prefixorder}
\end{figure}

For example, the interval labeled with $\Gamma_\rightarrow$ resp.\ $\Gamma_\leftrightarrow$ 
in the top part can be viewed as representing all communication patterns that start with the graph 
$\rightarrow$ resp.\ $\leftrightarrow$ in round 1. Note carefully that any two 
consecutive $\sigma$ and $\sigma'$ are indistinguishable for the process $p$ 
in-between in executions that start from the same input assignment. Overall, this imposes the order $\rightarrow \;
< \; \leftrightarrow \; < \; \leftarrow$, which can also be used for ordering the sets 
$\Gamma_\rightarrow < \Gamma_\leftrightarrow <  \Gamma_\leftarrow$.
By construction, every execution 
$\gamma_1$ with communication pattern $\G_1 \in \Gamma_\sigma$ and any execution
$\gamma_2$ with $\G_2 \in \Gamma_{\sigma'}$ with $I(\gamma_1)=I(\gamma_2)$ satisfy
$\dnonunif(\gamma_1,\gamma_2) \leq d_p(\gamma_1,\gamma_2) < 1/2$ in the $p$-view topology.
(Clearly, this holds for both $l$ and $r$ if $\G_1, \G_2 \in \Gamma_\sigma$ as well.)

The bottom part of our illustration is obtained by applying the ordering for $k=1$ to 
every edge in the top part. For example, $\Gamma_{\rightarrow\rightarrow}$ resp.\ $\Gamma_{\rightarrow\leftrightarrow}$
represents all communication patterns that start with the 2-prefix $\rightarrow\rightarrow$ 
resp.\ $\rightarrow\leftrightarrow$ sharing the same 1-prefix $\rightarrow$. Again,  $\rightarrow\rightarrow \; < \;
\rightarrow\leftrightarrow \;< \; \dots \; < \; \leftarrow\leftarrow$, as well as the sets 
$\Gamma_{\rightarrow\rightarrow} < \Gamma_{\rightarrow\leftrightarrow} < \dots < \Gamma_{\leftarrow\leftarrow}$,
are totally ordered, and any two executions $\gamma_1$ and $\gamma_2$ from direct successors 
$\Gamma_\sigma<\Gamma_{\sigma'}$
satisfy $\dnonunif(\gamma_1,\gamma_2) \leq d_p(\gamma_1,\gamma_2) < 1/4$ for some $p \in \Pi$.

For $\V=\{0,1\}$, four instances of these communication patterns, each representing one
of the 4 different possible input assignments, are connected in a cycle, as shown in 
\cref{fig:complete}. Thanks to this arrangement, the argument that any two executions $\gamma_1$ 
and $\gamma_2$ from direct successors $\Gamma_\sigma<\Gamma_{\sigma'}$
satisfy $\dnonunif(\gamma_1,\gamma_2) \leq d_p(\gamma_1,\gamma_2) < 1/4$ for some 
$p \in \Pi$ also extends to $\gamma_1$, $\gamma_2$ that are indistinguishable for a corner node.
This is particularly relevant for the strong decision sets $\Sigma_0$ and $\Sigma_1$ induced by 
the MinMax algorithm: According to our considerations above, there are only
two boundary points here, namely, the execution $\gamma_r \in \Sigma_1$ with 
$I_l(\gamma_r)=0$, $I_r(\gamma_r)=1$ and $\G = \{\leftarrow\}^\omega$ (corresponding to the
bottom-right black node in \cref{fig:complete}), and $\gamma_l \in \Sigma_1$ with
$I_l(\gamma_l)=1$, $I_r(\gamma_l)=0$ and $\G = \{\rightarrow\}^\omega$ (the
top-right white node in \cref{fig:complete}). $\Sigma_1$ comprises all executions
corresponding to the right edge in our figure, $\Sigma_0$ is made up by the executions
corresponding to all the other edges.

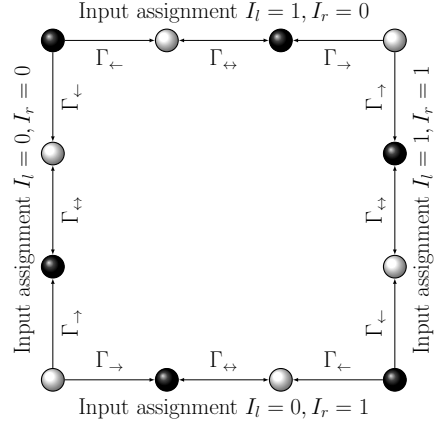
\begin{figure}
\begin{center}
  \scalebox{0.5}{
\begin{tikzpicture}
	\begin{pgfonlayer}{nodelayer}
		\node [style=VS] (7) at (-5, -2) {};
		\node [style=VW] (13) at (1, -2) {};
		\node [style=none] (22) at (-8, -1) {\Large $\Gamma_{\rightarrow}$};
		\node [style=none] (23) at (-2, -1) {\Large $\Gamma_{\leftrightarrow}$};
		\node [style=none] (24) at (4, -1) {\Large $\Gamma_{\leftarrow}$};
		\node [style=none] (25) at (-2, -3.5) {\LARGE Input assignment $I_l=0, I_r=1$};
		\node [style=VW] (26) at (-5, 16) {};
		\node [style=VS] (27) at (1, 16) {};
		\node [style=none] (30) at (-8, 15) {\Large $\Gamma_{\leftarrow}$};
		\node [style=none] (31) at (-2, 15) {\Large $\Gamma_{\leftrightarrow}$};
		\node [style=none] (32) at (4, 15) {\Large $\Gamma_{\rightarrow}$};
		\node [style=none] (33) at (-2, 17.5) {\LARGE Input assignment $I_l=1, I_r=0$};
		\node [style=VW] (34) at (7, 4) {};
		\node [style=VS] (35) at (7, 10) {};
		\node [style=VS] (36) at (7, -2) {};
		\node [style=VW] (37) at (7, 16) {};
		\node [style=none, rotate=90] (38) at (6, 1) {\Large $\Gamma_{\leftarrow}$};
		\node [style=none, rotate=90] (39) at (6, 7) {\Large $\Gamma_{\leftrightarrow}$};
		\node [style=none, rotate=90] (40) at (6, 13) {\Large $\Gamma_{\rightarrow}$};
		\node [style=none, rotate=90] (41) at (8.5, 7) {\LARGE Input assignment $I_l=1, I_r=1$};
		\node [style=VS] (42) at (-11, 4) {};
		\node [style=VW] (43) at (-11, 10) {};
		\node [style=VW] (44) at (-11, -2) {};
		\node [style=VS] (45) at (-11, 16) {};
		\node [style=none, rotate=90] (46) at (-10, 1) {\Large $\Gamma_{\rightarrow}$};
		\node [style=none, rotate=90] (48) at (-10, 13) {\Large $\Gamma_{\leftarrow}$};
		\node [style=none, rotate=90] (49) at (-12.5, 7) {\LARGE Input assignment $I_l=0, I_r=0$};
		\node [style=none, rotate=90] (50) at (-10, 7) {\Large $\Gamma_{\leftrightarrow}$};
	\end{pgfonlayer}
	\begin{pgfonlayer}{edgelayer}
		\draw [style=Einout] (7) to (13);
		\draw [style=Einout] (26) to (27);
		\draw [style=Einout] (34) to (35);
		\draw [style=Eout, in=-90, out=90] (36) to (34);
		\draw [style=Ein] (35) to (37);
		\draw [style=Einout] (42) to (43);
		\draw [style=Eout, in=-90, out=90] (44) to (42);
		\draw [style=Ein] (43) to (45);
		\draw [style=Eout, in=180, out=0] (44) to (7);
		\draw [style=Eout, in=180, out=0] (45) to (26);
		\draw [style=Ein] (27) to (37);
		\draw [style=Ein] (13) to (36);
	\end{pgfonlayer}
\end{tikzpicture}
  }
\end{center}
\caption{Complete representation of 1-prefixes in the LL model, for all input assignments.}
\label{fig:complete}
\end{figure}

\medskip

It is instructive to relate the situation to the problem of solving terminating
consensus in sub-models of LL. According to \cref{cor:consensusimpfair}, one
has to remove at least two fair or two pairs of unfair executions (see 
\cref{def:fairunfair}) from the connected space of admissible LL executions 
$\Sigma$ for this purpose. In particular, just excluding the above boundary
executions $\gamma_l$ and $\gamma_r$ from $\Sigma$ allows to solve 
terminating consensus: If e.g.\
$\{\leftarrow\}^\omega$ cannot occur, $l$ only needs to wait until the (guaranteed)
message from $r$ arrives and to decide on $r$'s input value then, whereas $r$ can
decide on its own input already initially. 

Topologically, this exclusion partitions the connected LL space $\Sigma$
(which is homeomorphic to a 1-sphere in the case of $\V=\{0,1\}$, see 
\cref{fig:complete}), into at least two connected (and hence clopen) 
components: Since the decision 
function of terminating consensus is continuous, this is mandatory for 
mapping the resulting space $\Sigma'$ to the discrete and hence disconnected
space $\V$. For stabilizing consensus, a separation into \emph{clopen} sets
is not needed: disjoint sets with included boundary points in between 
suffice here, since semi-continuous decision functions can map such sets to
a discrete space.



\subsubsection{Empty kernel impossibility}
\label{sec:emptykernel}

Among the few stabilizing consensus impossibility results known
so far is \cite[Thm.~6]{CM19:DC}, which states that any graph sequence 
with an empty kernel that is admissible for a message adversary makes
the problem unsolvable. Informally, the kernel of a communication
pattern $\G$ is the set of processes that can reach all other
processes (typically via multiple hops) infinitely often,
i.e., are broadcasters in every suffix of $\G$.

In the case of an empty kernel, there is a round $t_0$ from
which on no process reaches all other processes in $\Pi$.
Now consider the sub-model where all processes start only
at round $t_0$, which boils down to executions $\gamma$
that have a communication pattern $\G$ that starts with
a prefix of $t_0$ empty graphs and $bc(\gamma)=\emptyset$.
\cref{cor:broadcastablePS} asserts that stabilizing consensus
is impossible here.

\subsubsection{Bounded rootedness possibility}
\label{sec:boundedrootedness}

In \cite{CM19:DC}, Charron-Bost and Moran also introduced
a generalized version of the simple MinMax algorithm
stated in \cref{def:MinMax}, which they
called safe MinMax. Informally, it allows the processes
to take (i) the maximum over a suffix of (ii) the minumum
input value over a prefix, which both get longer and longer with increasing rounds.
In our setting, it can be expressed via the following decision function:
\begin{equation}
\Delta_p(C^t) = \max_{q \in \In_p(\theta_p(t)+1,t)} \bigl\{ \min_{s \in HO_q(C^{\theta_p(t)})} I_s(\gamma) \bigr\}\label{def:MinMaxGeneral}.
\end{equation}
Herein, $\theta_p(t)$ is a \emph{cut-off function}, computable by process $p$, which must
satisfy $\lim_{t\to\infty}\theta_p(t)=\lim_{t\to\infty} (t-\theta_p(t)) = \infty$.
The set $\In_p(\theta_p(t)+1,t)$ denotes the processes $q$ from which $p$ has
heared of by round $t$ when only considering messages that have been sent by $q$ in or after 
round $\theta_p(t)+1$.

It was proved that safe MinMax allows to solve stabilizing consensus
for every message adversary that satisfies bounded rootedness, which
implies that, in every execution $\gamma$, in every round $t$, there
is some broadcasting time $T \in \IN$ (possibly unknown to the processes) and a
non-empty set of broadcasters that reach every process by round
$t+T$.

Again, it is easy to show that \cref{def:MinMaxGeneral} leads to
the strong decision sets already used in \cref{sec:asyncwithcrashes}.
The proof is essentially the same, except that a more complex condition
regarding when the smaller input value $w < v$ is received by other
processes in $\delta_i'\in \Sigma_w$ occurring in the sequence 
$\delta_1',\delta_2', \dots \to \gamma \in \Sigma_v$ needs to be 
considered. More specifically, consider any round $t_i$ where $\theta_q(t_i)$ is large
enough such that \emph{every} process $q$ has already heared of at least once 
from every process it will ever hear of in $\delta_i'$, but small 
enough such that $q$ hears, by round $t_i$, about a message sent by at least one process 
$p'\in bc(\delta_i')$ (by assumption with $I_{p'}(\delta_i')=w$) after round $\theta_q(t_i)$. According
to \cite[Lem.~7]{CM19:DC}, this guarantees that $\Delta_q(t_i)=w$.
Since the broadcasting time $T$ is fixed and $\lim_{t\to\infty}\theta_q(t)=\lim_{t\to\infty} (t-\theta_q(t)) 
= \infty$, the latter can be guaranteed for any round $t\geq t_i$ as well.
For our desired sequence $\delta_1',\delta_2', \dots \to \gamma 
\in \Sigma_v$, choosing any non-decreasing sequence $t_i$ with $\lim_{i\to\infty}=\infty$
will do the job. 

For the limit sequence $\gamma$, this obviously implies that no 
process $q$ ever hears from $p'$, so $p'\not\in bc(\gamma)$ as needed.
As in the proof in \cref{sec:asyncwithcrashes}, we can hence choose a corresponding
sequence $\delta_1,\delta_2, \dots \to \gamma \in \Sigma_v$ that is identical
to $\delta_1',\delta_2', \dots$, except that $p'$ also has the initial
value $I_{p'}(\delta_i)=v > w$. Since it also converges to $\gamma$, this
again proves that $\gamma \in \bdin \Sigma_v$ is a limit point.

\subsubsection{Delayed lossy-link model impossibility}
\label{sec:DLLmodel}

The safe MinMax algorithm given in \cref{def:MinMaxGeneral} in
\cref{sec:boundedrootedness} also allows to solve stabilizing consensus
under the \emph{bounded delayed lossy-link} message adversary BDLL
defined as $\bigcup_{T\geq 0} \bigl( \bigcup_{k=0}^T \{\noarrow\}^k LL\bigr)^\omega$, which allows the classic LL model to be interleaved with
silence periods of at most $T$ rounds, for some unknown $T$, in every execution.

A natural question is whether restricting the maximum duration of
the silence periods in any execution to some fixed $T$ is mandatory
for solving stabilizing consensus. Since
the LL model guarantees a non-empty kernel in every execution, which
carries over to BDLL, this question is related to the more general
question posed in \cite{CM19:DC}: Is a non-empty kernel in every
execution, which is known to be necessary (recall \cref{sec:emptykernel})
for solving stabilizing consensus, also sufficient? Clearly,
safe MinMax would not work here, as one cannot guarantee that the
cut-off function $\theta_p(t)$ properly covers arbitrarily
finite silence periods, but there might be other algorithms.

In \cite{FR24:arxiv}, Felber and Rincon Galeana answered this
question negatively: They introduced the \emph{delayed lossy-link} message adversary
DLL defined as $\bigl(\{\noarrow\}^*LL\bigr)^\omega$, which allows arbitrary
but finite silence periods in every execution,
and showed that it is impossible to solve stabilizing consensus in this
model. We will use our topological characterization for providing an
alternative proof of this fact.

\medskip

Restricting our attention to binary stabilizing consensus,
i.e., $\V=\{0,1\}$, we assume for a contradiction that there
is a correct stabilizing consensus algorithm $\A$ for the DLL model.
Let $\Sigma_0$ and $\Sigma_1$ be the resulting decision sets,
which must be semi-open according to \cref{thm:Schar:nonunif}.
We will construct an admissible execution, which cannot be
assigned to any decision set, which provides the required
contradiction.

We start with using our stabilizing consensus algorithm $\A$ in
the LL model, which is of course a sub-model of DLL, cp.\
\cref{fig:DLL} and \cref{fig:prefixorder}. Let
$\hat{\Sigma}_0\subseteq \Sigma_0$ and $\hat{\Sigma}_1 \subseteq
\Sigma_1$ be the corresponding decision sets. From our topological
considerations in \cref{sec:LLmodel}, we know that there must be some
boundary point $\gamma$, w.l.o.g.\ $\gamma \in \bdin \hat{\Sigma}_0$,
which is also a limit point of $\hat{\Sigma}_1$. Note that $\gamma$
is hence also a limit point in the boundary of the DLL decision
sets $\Sigma_0$ and $\Sigma_1$. 

Unfortunately, however, all that we know about $\A$ is that it
produces semi-continuous decision sets. In particular, we do
not know anything about the set of boundary points $\A$ produces,
besides that $\bdin \hat{\Sigma}_v$ and $\hat{Y}=\bigcup_{v\in\V} \bdin \hat{\Sigma}_v$ are
nowhere dense by \cref{lem:nowheredenselimits}. Therefore, we cannot
just assume that there is some boundary point $\gamma$ that 
squarely separates $\hat{\Sigma}_0$ and $\hat{\Sigma}_1$,
but need a slightly more refined approach.

We start with some definitions: For an arbitray execution $\gamma$ with
graph sequence $\G$, let $\sigma=\G|_{k+1}$ be the $(k+1)$-prefix of
$\G$, for any $k\geq 0$. Recall that, for any execution $\beta$
with graph sequence $\B$, it holds that $\B|_{k+1}=\G|_{k+1}$
guarantees $\beta|_{k+1} \sim_{l,r} \gamma|_{k+1}$. 
Let $\lambda$ and $\rho$ be 
the (unfair) admissible executions based on the graph sequence $\sigma\{\rightarrow\}^\omega$ and 
$\sigma\{\leftarrow\}^\omega$, respectively. We call them \emph{LL border executions}, since, for every prefix 
size $k+1+\ell$, $\ell\geq0$, $\hat{\Sigma}_{\sigma\{\rightarrow\}^\ell}$ resp.\ 
$\hat{\Sigma}_{\sigma\{\leftarrow\}^\ell}$ are the smallest resp.\ largest element in 
the corresponding prefix-order, as illustrated in \cref{fig:prefixorder}. 
Obviously, both $\lambda, \rho \in B_{2^{-(k+1)}}(\gamma)$, where 
$B_{2^{-(k+1)}}(\gamma)$ denotes the open ball with radius $2^{-(k+1)}$ 
around $\gamma$ for the distance function $\dnonunif$. 

It follows from
the validity property (V) that the right and left corner executions $\gamma_r$
resp.\ $\gamma_l$ on the bottom edge in \cref{fig:complete}, which are based on the 
graph sequences $\G_r=(\leftarrow)^\omega$ resp.\ $\G_l=(\rightarrow)^\omega$,
satisfy $\gamma_r \in \hat{\Sigma}_1$ and $\gamma_l\in \hat{\Sigma}_0$. In every 
prefix-order, i.e., for arbitrary $(k+1)$-prefixes, $k \geq 0$, there must hence be at least
one prefix $\sigma$, of some execution $\gamma$, such that the corresponding LL border executions 
$\lambda$, $\rho$ satisfy $\lambda\in \hat{\Sigma}_0$ and $\rho \in \hat{\Sigma}_1$.
After all, somewhere, the decision value must flip from 0 to 1 when going
forward in the prefix-order. Again, this must also hold true in the DLL model, i.e.,
$\lambda\in \Sigma_0$ and $\rho \in \Sigma_1$, see \cref{fig:DLL} for an illustration
for $k=1$.

\begin{figure}[ht]
\begin{center}
\scalebox{0.8}{
\begin{tikzpicture}
	\begin{pgfonlayer}{nodelayer}
		\node [style=VS] (7) at (-5, -2) {};
		\node [style=VW] (13) at (1, -2) {};
		\node [style=VW] (14) at (-11, -2) {};
		\node [style=VS] (15) at (7, -2) {};
		\node [style=none] (22) at (-8, -1) {\footnotesize $\Gamma_{\rightarrow}$};
		\node [style=none] (23) at (-2, -1) {\footnotesize $\Gamma_{\leftrightarrow}$};
		\node [style=none] (24) at (4.5, -1) {\footnotesize $\Gamma_{\leftarrow}$};
		\node [style=VSsmall] (25) at (-9, -7) {};
		\node [style=VSsmall] (26) at (-5, -7) {};
		\node [style=VSsmall] (27) at (-1, -7) {};
		\node [style=VSsmall] (28) at (3, -7) {};
		\node [style=VSsmall] (29) at (7, -7) {};
		\node [style=VWsmall] (30) at (-11, -7) {};
		\node [style=VWsmall] (31) at (-7, -7) {};
		\node [style=VWsmall] (32) at (-3, -7) {};
		\node [style=VWsmall] (33) at (1, -7) {};
		\node [style=VWsmall] (34) at (5, -7) {};
		\node [style=none] (45) at (-10, -6) {\footnotesize $\Gamma_{\rightarrow\rightarrow}$};
		\node [style=none] (46) at (-8, -6) {\footnotesize $\Gamma_{\rightarrow\leftrightarrow}$};
		\node [style=none] (47) at (-6, -6) {\footnotesize $\Gamma_{\rightarrow\leftarrow}$};
		\node [style=none] (48) at (0, -6) {\footnotesize $\Gamma_{\leftrightarrow\rightarrow}$};
		\node [style=none] (49) at (-2, -6) {\footnotesize $\Gamma_{\leftrightarrow\leftrightarrow}$};
		\node [style=none] (50) at (-4, -6) {\footnotesize $\Gamma_{\leftrightarrow\leftarrow}$};
		\node [style=none] (51) at (6, -6) {\footnotesize $\Gamma_{\leftarrow\leftarrow}$};
		\node [style=none] (52) at (4, -6) {\footnotesize $\Gamma_{\leftarrow\leftrightarrow}$};
		\node [style=none] (53) at (2, -6) {\footnotesize $\Gamma_{\leftarrow\rightarrow}$};
		\node [style=none] (54) at (-2, -4) {};
		\node [style=none] (55) at (-2, -5) {\footnotesize $\Gamma_{\noarrow}$};
		\node [style=none] (56) at (-8, -9) {};
		\node [style=none] (57) at (-2, -9) {};
		\node [style=none] (58) at (4, -9) {};
		\node [style=none] (59) at (-8, -9.75) {\footnotesize $\Gamma_{\rightarrow\noarrow}$};
		\node [style=none] (60) at (-2, -9.75) {\footnotesize $\Gamma_{\leftrightarrow\noarrow}$};
		\node [style=none] (61) at (4, -9.75) {\footnotesize $\Gamma_{\leftarrow\noarrow}$};
		\node [style=T] (84) at (-5, -18) {};
		\node [style=T] (86) at (-2, -18) {};
		\node [style=T] (87) at (1, -18) {};
		\node [style=none] (88) at (-5, -19) {$\rho$};
		\node [style=none] (89) at (1, -19) {$\lambda$};
		\node [style=none] (90) at (-2, -19) {$\gamma$};
		\node [style=none] (100) at (-11, -1) {};
		\node [style=none] (101) at (-11, -20) {};
		\node [style=none] (102) at (-5, -1) {};
		\node [style=none] (103) at (-5, -20) {};
		\node [style=none] (104) at (1, -1) {};
		\node [style=none] (105) at (7, -1) {};
		\node [style=none] (106) at (1, -20) {};
		\node [style=none] (107) at (7, -20) {};
		\node [style=VSsmall] (108) at (-5, -12) {};
		\node [style=VWsmall] (109) at (1, -12) {};
		\node [style=none] (110) at (-8, -12) {\footnotesize $\Gamma_{\noarrow\rightarrow}$};
		\node [style=none] (111) at (-2, -13) {\footnotesize $\Gamma_{\noarrow\leftrightarrow}$};
		\node [style=none] (112) at (4, -12) {\footnotesize $\Gamma_{\noarrow\leftarrow}$};
		\node [style=none] (113) at (-2, -15) {};
		\node [style=none] (114) at (-2, -16) {\footnotesize $\Gamma_{\noarrow\noarrow}$};
	\end{pgfonlayer}
	\begin{pgfonlayer}{edgelayer}
		\draw [style=Einout] (7) to (13);
		\draw [style=Eout, in=180, out=0] (14) to (7);
		\draw [style=Ein] (13) to (15);
		\draw [style=Eout] (30) to (25);
		\draw [style=Eout] (33) to (28);
		\draw [style=Ein] (31) to (26);
		\draw [style=Ein] (34) to (29);
		\draw [style=Einout] (25) to (31);
		\draw [style=Einout] (32) to (27);
		\draw [style=Einout] (28) to (34);
		\draw [style=Eout] (26) to (32);
		\draw [style=Ein] (27) to (33);
		\draw [bend right=15] (14) to (54.center);
		\draw [bend right=15] (54.center) to (15);
		\draw [bend right, looseness=1.25] (30) to (56.center);
		\draw [bend left, looseness=1.25] (26) to (56.center);
		\draw [bend right, looseness=1.25] (26) to (57.center);
		\draw [bend left, looseness=1.25] (33) to (57.center);
		\draw [bend right, looseness=1.25] (33) to (58.center);
		\draw [bend right, looseness=1.25] (58.center) to (29);
		\draw [style=Enoned] (101.center) to (100.center);
		\draw [style=Enoned] (103.center) to (102.center);
		\draw [style=Enoned] (106.center) to (104.center);
		\draw [style=Enoned] (107.center) to (105.center);
		\draw (84) to (87);
		\draw [style=Eout, bend right] (30) to (108);
		\draw [style=Einout] (108) to (109);
		\draw [style=Eout, bend left] (29) to (109);
		\draw [bend right=45, looseness=1.25] (30) to (113.center);
		\draw [bend left=45, looseness=1.25] (29) to (113.center);
	\end{pgfonlayer}
\end{tikzpicture}
}
\end{center}
\caption{Representation of all 1-prefixes and 2-prefixes in the DLL model, for a fixed input assignment.}
\label{fig:DLL}
\end{figure}
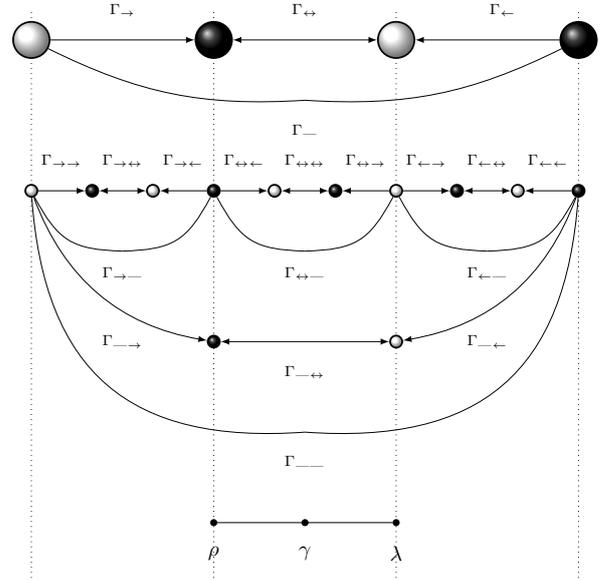

For any $m\geq0$, we now define the admissible parametrized DLL execution $\lambda(m)$, which is
based on the graph sequence $\sigma\{\noarrow\}^m\{\rightarrow\}^\omega$, and $\rho(m)$, which
is based on $\sigma\{\noarrow\}^m\{\leftarrow\}^\omega$. Besides $\lambda(0)=\lambda$ and
$\rho(0)=\rho$, they satisfy 
\begin{align}
   d_l(\lambda(m),\lambda)&=0\label{eq:dlambda},\\
   d_r(\rho(m),\rho)&=0\label{eq:drho}.
\end{align}
Moreover, since $\lambda(m)$ and $\rho(m)$ start from the same prefix $\sigma$,
  \begin{equation}
    \lambda(m)|_{k+m} \sim_{l,r} \rho(m)|_{k+m}.
    \label{eq:lamrho}
     \end{equation}
  Due to \cref{eq:lamrho}, we find that, for any $m\geq 0$, the respective
  output values after round $k+m$ satisfy
\begin{align}
  O_l\bigl(\lambda(m)|_{k+m}\bigr)&=O_l\bigl(\rho(m)|_{k+m}\bigr)\label{eq:lamrhosamel},\\
    O_r\bigl(\lambda(m)|_{k+m}\bigr)&=O_r\bigl(\rho(m)|_{k+m}\bigr)\label{eq:lamrhosamer}.
\end{align}
However, since $\lambda\in \Sigma_0$ and $d_l(\lambda,\lambda(m))=0$ by \cref{eq:dlambda}, 
there is a round $k_l\geq k$ after which $O_l=0$ in both $\lambda$ and $\lambda(m)$. Analogously, since
$\rho\in\Sigma_1$, there is a round $k_r\geq k$ after which $O_r=1$ in both $\rho$ and $\rho(m)$. For  $m \geq \max\{k_l,k_r\}-k$,
\cref{eq:lamrhosamel} resp.\ \cref{eq:lamrhosamer} ensures
that $O_l=0$ also in $\rho(m)$ resp.\ $O_r=1$ also in $\lambda(m)$.
Consequently, the Stabilizing Agreement property (SA) is violated both in
$\lambda(m)$ and $\rho(m)$ in round $k+m$, so the executions
cannot have stabilized by round $k+m$. In the notation of \cite{FR24:arxiv},
any such execution has a \emph{conflicting} prefix. Since both $k\geq 0$
and $m \geq \max\{k_l,k_r\}-k$ can be chosen arbitrarily large, there
must be at least one non-stabilizing execution $\gamma$ that cannot be
uniquely assigned to either $\Sigma_0$ or $\Sigma_1$; the stipulated
stabilizing consensus algorithm $\A$ hence cannot be correct.

\subsubsection{One-message lossy link possibility}
\label{sec:onemessage}

As our final example for demonstrating the explanatory power of our 
topological approach, we consider some apparently minor strengthening 
of the DLL model, which nevertheless makes stabilizing consensus
trivially solvable. We call it the \emph{one-message lossy link} model, 
where the set of allowed graph sequences is the union of
$\bigl\{\{\noarrow\}^*\rightarrow\{\noarrow\}^\omega\bigr\}$ and
$\bigl\{\{\noarrow\}^*\leftarrow\{\noarrow\}^\omega\bigr\}$. Obviously, in
any admissible execution, exactly one message (either $\rightarrow$ or
$\leftarrow$) is successfully received.

Despite the infinite suffix $\{\noarrow\}^\omega$, it is trivial to 
solve stabilizing consensus in this model: Initially,
$l$ and $r$ set $O_l=I_l$ and $O_r=I_r$ and keep that value, unless
on of them, $x \in \{l,r\}$, receives a message containing the
other's current output value $O$. If so, $x$ sets $O_x=O$ and keeps
that value. By contrast, it will turn out that terminating consensus 
cannot be solved in this model.

Our topological characterization easily reveals what happens here. It
suffices to characterize the decision sets $\Sigma_0$ and $\Sigma_1$
imposed by this algorithm, in the restricted setting where 
$I_l=0$ and $I_r=1$ are the same in all admissible
executions; therefore, we can consider admissible executions and graph
sequences as being equivalent. Generalizing our considerations to proper
stabilizing consensus algorithms is straightforward. 

For $i \geq 0$, let $\eta=\{\noarrow\}^\omega$, 
$\alpha_i = \{\noarrow\}^i\rightarrow\{\noarrow\}^\omega$, and
$\beta_i=\{\noarrow\}^i\leftarrow\{\noarrow\}^\omega$. Then, $\Sigma_0=\{\alpha_i\mid i \geq 1\}$ and $\Sigma_1=\{\beta_i\mid i \geq 1\}$. Moreover,
for any $i\geq 1, k\geq 1$, we find $d_l(\alpha_i,\alpha_k)=0$ and
$d_l(\alpha_i,\eta)=0$ (since $l$ never gets any message in any
of these executions), as well as $d_r(\beta_i,\beta_k)=0$ and
$d_r(\beta_i,\eta)=0$. On the other hand, $d_r(\alpha_i,\alpha_{i+k})=2^{-(i+1)}$
and $d_r(\alpha_i,\eta)=2^{-(i+1)}$, as well as $d_r(\beta_i,\beta_k)=2^{-(i+1)}$ and $d_l(\beta_i,\eta)=2^{-(i+1)}$. Since the triangle inequality reveals
$d_l(\alpha_i,\beta_k) \leq
d_l(\alpha_i,\eta) + d_l(\eta,\beta_k) = d_l(\eta,\beta_k)=2^{-(k+1)}$ and
$d_r(\alpha_i,\beta_k) \leq
d_r(\alpha_i,\eta) + d_r(\eta,\beta_k) = d_r(\eta,\alpha_k)=2^{-(k+1)}$,
it turns out that any $\alpha_i$ is a limit point of the sequence
$\beta_1,\beta_2, \dots$ in the $l$-view topology. In addition,
trivially, any $\alpha_i$ is also a limit point of $\alpha_1,\alpha_2, \dots$ in the $l$-view topology.
Analogously,  any $\beta_i$ is a limit point of both the sequence
$\alpha_1,\alpha_2, \dots$ and $\beta_1,\beta_2,\dots$ in the $r$-view
topology. Therefore, \cref{cor:consensusseparation} implies that terminating
consensus is impossible in this model.

As a consequence, in the non-uniform topology, the diameter
(see \cref{def:diameterset}) of both $\Sigma_0$ and $\Sigma_1$ is 0.
Moreover, $\bdin \Sigma_0 = \Sigma_0$ and $\bdin \Sigma_1 = \Sigma_1$,
whereas $\bdtwo \bdin \Sigma_0 = \Sigma_1$ and $\bdtwo \bdin \Sigma_1 =
\Sigma_0$ since both $\Sigma_0$ and $\Sigma_1$ are dense in $\Sigma$,
i.e., $\cl{\Sigma}_0 =  \cl{\Sigma}_1 = \Sigma$. It follows that 
the conditions for \cref{thm:Schar:nonunif} are satisfied, such
that stabilizing consensus is indeed solvable in this model. Note that
case (c) in the proof of \cref{thm:Schar:nonunif}, i.e., \cref{eq:SCcase3}, 
cannot occur here at all, since $\bdtwo \bdin \Sigma_0 \cap \bdtwo \bdin \Sigma_1 = \emptyset$.

\section{Conclusions}
\label{sec:conclusions}

We provided a complete characterization of the solvability/impossibility of 
deterministic stabilizing consensus in any computing model with benign
process and communication faults using point-set topology. Using the 
topologies for infinite executions introduced in \cite{NSW24:JACM} 
for terminating consensus, we proved that semi-open
decision sets and semi-continuous decision functions as introduced 
in \cite{Lev63} are the appropriate means for this characterization. 
We also showed that multi-valued stabilizing consensus with weak and strong validity 
are equivalent, like for terminating consensus. We demonstrated the power of
our topological characterization by easily applying it to (variants of) all the 
known possibility and impossibility results for stabilizing consensus known 
so far.

\backmatter

\bmhead{Acknowledgements}

This work has been supported by the Austrian Science Fund (FWF) under project ByzDEL ((10.55776/P33600) and DMAC ((10.55776/P32431).




\bibliography{localbib,lit}

\end{document}